\documentclass[letterpaper]{article} 
\usepackage{aaai24}  
\usepackage{times}  
\usepackage{helvet}  
\usepackage{courier}  
\usepackage[hyphens]{url}  
\usepackage{graphicx} 
\usepackage{natbib}  
\usepackage{caption} 
\usepackage{algorithm}
\usepackage{algorithmic}

\usepackage{newfloat}
\usepackage{listings}
\DeclareCaptionStyle{ruled}{labelfont=normalfont,labelsep=colon,strut=off} 
\floatstyle{ruled}
\newfloat{listing}{tb}{lst}{}
\floatname{listing}{Listing}
\usepackage{microtype}
\usepackage{graphicx}
\usepackage{subfigure}
\usepackage{booktabs} 

\usepackage{amsmath}
\usepackage{amssymb}
\usepackage{mathtools}
\usepackage{amsthm}

\usepackage{bm}

\usepackage[capitalize,noabbrev]{cleveref}

\theoremstyle{plain}
\newtheorem{theorem}{Theorem}

\theoremstyle{definition}
\newtheorem{definition}[theorem]{Definition}
\newtheorem{assumption}[theorem]{Assumption}
\theoremstyle{remark}

\newtheorem{apptheorem}{Theorem}[section]
\newtheorem{appproposition}[apptheorem]{Proposition}
\newtheorem{applemma}[apptheorem]{Lemma}
\newtheorem{appcorollary}[apptheorem]{Corollary}
\theoremstyle{definition}
\newtheorem{appdefinition}[apptheorem]{Definition}
\newtheorem{appassumption}[apptheorem]{Assumption}
\theoremstyle{remark}
\newtheorem{appremark}[apptheorem]{Remark}
\newtheorem{appexample}[apptheorem]{Example}

\usepackage[textsize=tiny]{todonotes}
\usepackage{enumitem}

\title{Iterative Regularization with $k$-support Norm: \\ An Important
  Complement to Sparse Recovery}
  \author{
    William de Vazelhes\textsuperscript{\rm 1},
    Bhaskar Mukhoty\textsuperscript{\rm 1},
    Xiao-Tong Yuan\textsuperscript{\rm 2},
    Bin Gu\textsuperscript{\rm 1,3}\thanks{Corresponding author}
}
\affiliations{
    \textsuperscript{\rm 1}MBZUAI, Abu Dhabi, UAE \\
    \textsuperscript{\rm 2}Nanjing University, Suzhou, China \\
    \textsuperscript{\rm 3}Jilin University, Changchun, China\\
    \{wdevazelhes,bhaskar.mukhoty,xtyuan1980,jsgubin\}@gmail.com
}

\usepackage{bibentry}

\begin{document}



\usetikzlibrary{arrows,fit,positioning,shapes,snakes}

\newdimen\arrowsize
\pgfarrowsdeclare{arcsq}{arcsq}
{
 \arrowsize=0.2pt
 \advance\arrowsize by .5\pgflinewidth
 \pgfarrowsleftextend{-4\arrowsize-.5\pgflinewidth}
 \pgfarrowsrightextend{.5\pgflinewidth}
}
{
 \arrowsize=1.5pt
 \advance\arrowsize by .5\pgflinewidth
 \pgfsetdash{}{0pt} 
 \pgfsetroundjoin   
 \pgfsetroundcap    
 \pgfpathmoveto{\pgfpoint{0\arrowsize}{0\arrowsize}}
 \pgfpatharc{-90}{-140}{4\arrowsize}
 \pgfusepathqstroke
 \pgfpathmoveto{\pgfpointorigin}
 \pgfpatharc{90}{140}{4\arrowsize}
 \pgfusepathqstroke
}

\newcommand{\E}{\mathbb{E}}
\newcommand{\rs}[2]{\binom{[#2]}{#1}}  
\newcommand{\R}{\mathbb{R}}
\newcommand{\prox}{\text{prox}}
\newcommand{\sus}{s_2}  
\newcommand{\bmx}{\bm{x}}
\newcommand{\bmb}{\bm{b}}
\newcommand{\p}{\bm{p}}
\newcommand{\q}{\bm{q}}
\newcommand{\g}{\bm{g}}
\newcommand{\X}{\bm{X}}
\newcommand{\y}{\bm{y}}
\newcommand{\z}{\bm{z}}
\newcommand{\x}{\bm{x}}
\newcommand{\w}{\bm{w}}
\newcommand{\uu}{\bm{u}}
\newcommand{\vv}{\bm{v}}
\newcommand{\ddelta}{\bm{\delta}}
\newcommand{\eepsilon}{\bm{\epsilon}}

\newcommand{\XX}{\bm{X}}
\newcommand{\UU}{\bm{U}}
\newcommand{\DD}{\bm{D}}
\newcommand{\VV}{\bm{V}}

\newcommand{\e}{\bm{\varepsilon}}
\newcommand{\down}{^\downarrow}
\newcommand{\stoch}{\xi}
\newcommand{\OO}{\mathcal{O}}
\newcommand{\com}[1]{\textcolor{green}{#1}}
\newcommand{\wil}[1]{\textcolor{blue}{#1}}
\newcommand{\beg}[1]{\textcolor{black}{#1}}
\definecolor{amber}{rgb}{1.0, 0.75, 0.0}
\definecolor{ao}{rgb}{0.0, 0.5, 0.0}

\maketitle

\begin{abstract}
Sparse recovery is ubiquitous in machine learning and signal processing. Due to the NP-hard nature of sparse recovery, existing methods are known to suffer either from restrictive (or even unknown) applicability conditions, or high computational cost. Recently, iterative regularization methods have emerged as a promising fast approach because they can achieve sparse recovery in one pass through early stopping, rather than the tedious grid-search used in the traditional methods.
However, most of those iterative methods are based on the $\ell_1$ norm which requires restrictive applicability conditions and could fail in many cases. Therefore, achieving sparse recovery with iterative regularization methods under a wider range of conditions has yet to be further explored.
To address this issue, we propose a novel iterative regularization algorithm, IRKSN, based on the $k$-support norm regularizer rather than the $\ell_1$ norm. We provide conditions for sparse recovery with IRKSN, and compare them with traditional conditions for recovery with $\ell_1$ norm regularizers. Additionally, we give an early stopping bound on the model error of IRKSN with explicit constants, achieving the standard linear rate for sparse recovery. Finally, we illustrate the applicability of our algorithm on several experiments, including a support recovery experiment with a correlated design matrix.
\end{abstract}

\section*{Introduction}

\beg{Sparse recovery is ubiquitous in machine learning and signal processing, with applications ranging from single pixel camera, to MRI, or radar}\footnote{An introduction to this topic, as well as an extensive review of its applications  can be found in \cite{foucart2013} and \cite{Wright2022}.}.  In particular, with the ever-increasing amount of information, real-life datasets often contain much more features than samples: this is for instance the case in DNA microarray datasets \cite{golub1999molecular}, text data \cite{lang1995newsweeder}, or image data such as fMRI \cite{belilovsky15b}, where the number of features is generally  much larger than the number of samples. In these high-dimensional settings, finding a linear model is under-specified, and therefore, one often needs to leverage additional assumptions about the true model, such as sparsity, to recover it. Usually, the problem is formulated as follows: we seek to recover a sparse vector $\w^* \in \mathbb{R}$ from its noisy linear measurements
$$\y^{\delta} = \XX \w^{*} + \eepsilon $$
Here, $\y^{\delta}$ is a noisy measurement vector, i.e. a noisy version of the true target vector $\bm{y} =  \XX \w^{*}$, $\XX = \left[ \bmx_1, ..., \bmx_d\right] \in \R^{n \times d}$ is a measurement matrix, also called design matrix, $\eepsilon \in \R^n$ is some bounded noise ($\| \eepsilon \|_2 \leq \delta $, with $\delta \in \mathbb{R}_{+}$), and $\w^*$ is the unknown $k$-sparse vector, i.e. containing only $k$ non-zero components, that we wish to estimate with a vector $\hat{\w}$ obtained by running some sparse recovery algorithm on observations $\y^{\delta}$ and $\XX$. Unfortunately, this problem is NP-hard in general, even in the noiseless setting \cite{natarajan1995sparse}.

However, most of those iterative methods are based on the $\ell_1$ norm which requires restrictive applicability conditions and could fail in many cases. We discuss such related works in more details in the next section. Therefore, achieving sparse recovery with iterative regularization methods under a wider range of conditions has yet to be further explored.

\beg{To address this issue, we propose a novel iterative regularization algorithm, IRKSN, based on the $k$-support norm regularizer rather than the $\ell_1$ norm.} That norm was first introduced in \cite{argyriou2012sparse}, as a way to improve upon the ElasticNet for sparse prediction. More precisely, we plug the $k$-support norm regularizer, for which there exist efficient proximal computations \cite{argyriou2012sparse,mcdonald2016new}, into the primal-dual framework for iterative regularization described in \cite{matet2017}.

\beg{We then provide some conditions for sparse recovery with IRKSN, and discuss on a simple example how they compare with traditional conditions for recovery with $\ell_1$ norm regularizers.}

More precisely, we elaborate on why such specific conditions include cases that are not included in some usual sufficient conditions for recovery with traditional methods based on the $\ell_1$ norm (see Figure \ref{fig:imp}) (we describe such conditions for recovery with $\ell_1$ norm in more details in Assumption~\ref{ass:assl1}). Since those types of conditions are still slightly opaque to interpret, we do as is common in the literature (such as in \cite{zou2005regularization,jia2010}), namely, we discuss and compare those solutions with the help of an illustrative example. We also give an early stopping bound on the model error of IRKSN with explicit constants, achieving the standard linear rate for sparse recovery. 

 Finally, we illustrate the applicability of IRKSN on several experiments, including a support recovery experiment with a correlated design matrix, and show that it allows to identify the support more accurately than its competitors.

\paragraph{Contributions.} We summarize the main contributions of our paper as follows:
\begin{enumerate}
\item  We introduce a new algorithm, IRKSN, which allows recovery of the true sparse vector under conditions for which 
some sufficient conditions for recovery with $\ell_1$ norm do not hold. We discuss the difference between those conditions on a detailed example.
\item We give an early stopping bound on the model error of IRKSN with explicit constants, achieving the standard linear rate for sparse recovery.
\item We illustrate the applicability of our algorithm on several experiments, including a support recovery experiment with a correlated design matrix, and show that it allows support recovery with a higher F1 score than its competitors.
\end{enumerate}

\begin{table*}[ht]
  \begin{center}
    \begin{small}
      \begin{sc}
        \begin{tabular}{p{5.5cm}p{5.5cm}p{3cm}p{1.8cm}}
  \toprule
  Method & Condition on $\XX$& Bound on $\| \hat{\w} - \w^* \|$ &  Complexity \\
  \midrule
  IHT \cite{blumensath2009}& RIP &$O(\delta)$ & $O(T)$\\
  Lasso \cite{tibshirani1996}  &  $\max\limits_{\ell \in \bar{S}}|\langle\XX_S^{\dagger} \bmx_\ell, \operatorname{sgn}(\w^*_S)\rangle|<1 ^{(2)}  $ & $O(\delta)$  & $O(\Lambda T )$  \\
  ElasticNet \cite{zou2005regularization}  & - & - & $O(\Lambda T )$\\
  KSN pen. \cite{argyriou2012sparse}  & -  & - & $O(\Lambda T )$\\
  OMP \cite{tropp2007} & RIP & $O(\delta)$ & $O(k)$\\
  SRDI \cite{osher2016} &      $\left\{   \begin{array}{l}\exists \gamma \in (0, 1]: ~ \XX_S^{\top}\XX_{S} \geq n \gamma I_{d, d}\\
                                          \exists \eta \in (0, 1): ~\|\XX_{\bar{S}}\XX_{S}^{\dagger}\|_{\infty} \leq 1 - \eta
                                        \end{array}\right. $        &    $ O(\sigma \sqrt{\frac{k \log d}{n}}) ~^{(1)}
                                        $     &        $O(T)$      \\
                                        IROSR \cite{vaskevicius2019} & RIP & $O(\sigma \sqrt{\frac{k \log d}{n}})~^{(1)}$ &$O(T )$ \\
                                        IRCR \cite{molinari2021iterative} &  $\max\limits_{\ell \in \bar{S}}|\langle \XX_S^{\dagger} \bmx_\ell, \operatorname{sgn}(\w_S^*)\rangle |<1 ^{(2)} $ & $O(\delta)$ &$O( T )$ \\
                                        \textbf{IRKSN (ours)} & $\max\limits_{\ell \in \bar{S}} | \langle \XX_S^{\dagger} \bmx_{\ell}, \w^*_{S}\rangle |< \min\limits_{j \in S}| \langle \XX_S^{\dagger} \bmx_{j}, \w^*_{S}\rangle | $ &  $O(\delta)$ & $O( T  )$\\
                                        \bottomrule
                                      \end{tabular}
                                    \end{sc}
                                  \end{small}
                                \end{center}
                                \caption{Comparison of the existing algorithms for sparse recovery in the literature, including conditions on $\XX$ and $\w^*$ sufficient for recovery. $T$ is the number of iterations each algorithm is ran for, and $\Lambda$ is the number of values of $\lambda$ that need to be tried out (for penalized methods). \quad  $^{(1)}$ assuming $\epsilon \underset{i.i.d.}{\sim} \mathcal{N}(0, \sigma^2)$.  \quad $^{(2)}$: Additionally, $\XX_S$ should be injective.}
                                \label{tab:pen}
                              \end{table*}

\section*{Preliminaries}
\paragraph{Notations.} We first recall a few definitions and notations used in the rest of the paper. We denote all vectors and matrices variables in bold font. For $S \subseteq [d]$, $\bar{S}$ denotes $[d] \setminus S$. For any matrix $\bm{M} \in \R^{n \times d}$, $\bm{m}_i$ denotes its $i$-th column for $i \in \mathbb{N}$, $\bm{M}^{\top}$  its transpose, $\bm{M}^{\dagger}$  its Moore-Penrose pseudo-inverse \cite{golub2013matrix}, $\| \bm{M} \|$  its nuclear norm, and $\bm{M}_S$ its column-restriction to a support $S \subseteq [d]$, i.e. the $n \times |S|$ matrix composed of the $|S|$ columns of $\bm{M}$ of indices in $S$. For a vector $\w \in \R^d$, $\text{supp}(\w)$ denotes its support $\w$, that is, the coordinates of the non-zero components of $\w$, $w_i$ denotes its $i$-th component, $|w|\down_i$ denotes its $i$-th top absolute value, and $\| \w\|$ denotes its $\ell_2$ norm.

\begin{figure}[hbt!]
  \centering
\includegraphics[scale=0.4]{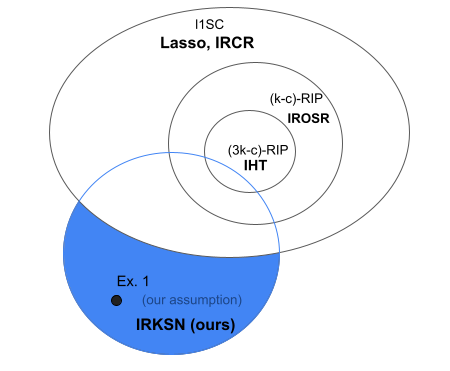}
\caption{Conditions for recovery in various settings: l1SC corresponds to the condition $\max_{\ell \in \bar{S}}|\langle\XX_S^{\dagger} \bmx_{\ell}, \operatorname{sgn}(\w^*_S)\rangle|<1$. ``ours'' denotes the condition  $\max_{i \in \bar{S}}| \langle \XX_S^{\dagger} \bmx_{i}, \w^*_{S}\rangle| < \min_{j \in S}| \langle \XX_S^{\dagger} \bmx_{j}, \w^*_{S}\rangle| $. $c$ denotes some constant in $[0, 1]$. Here $3k$-RIP is shown for indicative purposes, corresponding to the condition for IHT as described in \cite{blumensath2009}. As we can see, for some cases (in blue), only IRKSN (our algorithm) can provably ensure sparse recovery.}
   \label{fig:imp}
 \end{figure}

More generally $\| \w \|_p$ denotes its $\ell_p$ norm for $p \in [1, +\infty)$, and $\|\w\|_0$ denotes its number of non-zero components. $\w_S \in \R^k$ denotes its restriction to a support $S$ of size $k$, that is, the sub-vector of size $k$ formed by extracting only the components $w_i$ with $i \in S$.  $\text{sgn}(\w)$ denotes the vector of its signs (with the additional convention that if $\w_i = 0$, $\text{sgn}(\w)_i=0$).

\paragraph{Related works.}  

\beg{Due to the NP-hard nature of sparse recovery, existing methods are known to suffer either from restrictive (or even unknown) applicability conditions, or high computational cost.} Amongst those methods, a first group of methods can achieve an exact sparsity $k$ of the estimate $\hat{\w}$: Iterative Hard Thresholding \cite{blumensath2009} returns an estimate $\hat{\w}$ which recovers $\w^*$ up to an error $\| \hat{\w} - \w^*\| \leq O (\delta)$, if the design matrix $\XX$ satisfies some
Restricted Isometry Property (RIP) \cite{blumensath2009}.
However, as mentioned in \cite{Jain14}, this condition is very restrictive, and does not hold in  most high-dimensional problems.
Greedy methods, such as Orthogonal Matching Pursuit (OMP) \cite{tropp2007}, also can return an exactly $k$-sparse vector, and bounds on the recovery of a (generalized version of) OMP, of the type $\| \hat{\w} - \w^*\| \leq O (\delta)$, can be found for instance in \cite{wang2015recovery}, under some RIP condition.

A second set of methods for sparse recovery solve the following penalized problem:
$$ (P): \min_{\w} \| \XX \w - \y^{\delta} \|^2 + \lambda R(\w)$$
Where $R$ is a regularizer, such as the $\ell_1$ norm as is done in the Lasso method \cite{tibshirani1996}, and $\lambda$ is a penalty parameter that needs to be tuned. For a given $\lambda$,  $(P)$ is usually solved through a convex optimization algorithm, and returns a solution $\hat{\w}$ of $(P)$, as an estimate of $\w^*$. Amongst those, one of the most important algorithms for sparse recovery, the Lasso \cite{tibshirani1996}, has been proven in \cite{grasmair2011}  to give a bound $\|\hat{\w} - \w^*\| \leq O (\delta)$  under the so-called \textit{source conditions} (described in Condition 4.3 from \cite{grasmair2011}) which are implied by the following more intuitive conditions: $\XX_S$ is injective, and $\max_{\ell \in \bar{S}}|\langle\XX_S^{\dagger} \bmx_\ell, \operatorname{sgn}(\w^*_S)\rangle|<1  $ (we detail this implication in Assumption~\ref{ass:assl1}). 
Following the Lasso, the ElasticNet was later developed to solve the problem of a design matrix with possibly high correlations. However, although some conditions for \textit{statistical consistency} exist for the ElasticNet \cite{jia2010}, to the best of our knowledge, there is no model error bound (and conditions thereof) for recovery with ElasticNet. Finally, the $k$-support norm regularization has also been used successfully as a penalty \cite{argyriou2012sparse}, with even better empirical results than the ElasticNet, but no explicit error bounds on model error (and the conditions thereof) currently exists: indeed, their work was mostly focused on \textit{sparse prediction} and not \textit{sparse recovery}. Efficient solvers have later been derived for the Lasso using for instance coordinate descent and its variants \cite{fang2020,bertrand2021anderson}. However, even with efficient solvers, these penalized methods need to tune the parameter $\lambda$, which is very costly.

\beg{Recently, iterative regularization methods have emerged as a promising fast approach because they can achieve sparse recovery in one pass through early stopping, rather than the tedious grid-search used in traditional methods}. They solve the following problem
\begin{align*}
 (I): \quad \quad &\min_{\w} R(\w)\\
\text{s.t.} \quad & \XX \w = \y^{\delta}  
\end{align*}
An iterative algorithm is used to solve it, and returns some $\hat{\w}$ to estimate $\w^*$. Importantly, $\hat{\w}$ is obtained by stopping the algorithm before convergence, also called \textit{early stopping}. One of the first amongst these methods, SRDI \cite{osher2016}, achieves a rate of $\| \hat{\w} - \w\| \leq O(\sigma \sqrt{\frac{k \log d}{n}})$ with high probability, assuming $\epsilon \underset{i.i.d.}{\sim} \mathcal{N}(0, \sigma)$, and two conditions: (1) $\exists \gamma \in (0, 1]: ~ \XX_S^{\top}\XX_{S} \geq n \gamma I_{d, d}$ (Restricted Strong Convexity) and (2) $\exists \eta \in (0, 1): ~\|\XX_{\bar{S}}\XX_{S}^{\dagger}\|_{\infty} \leq 1 - \eta$.
IROSR \cite{vaskevicius2019} uses an iterative regularization scheme that is based on a reparameterization of the problem (I). They prove a high probability model consistency bound of $\|\hat{\w} - \w^*\| \leq O (\sigma\sqrt{\frac{k \log d}{n}})$, assuming  the ($(k+1, c)$-RIP for some constant $c(k, \w^*, \XX, \epsilon)$.  Similar to their work is \cite{zhao2022}: under similar conditions, they also obtain a similar rate. Finally, \cite{molinari2021iterative} provide bounds of the form $\| \hat{\w} - \w\| \leq O(\delta)$, under the same \textit{source conditions} as in \cite{grasmair2011}.

\beg{However, most of those iterative methods are based on the $\ell_1$ norm which requires restrictive applicability conditions and could fail in many cases. } Indeed, in those cases, the conditions for recovery with the methods described above 
(e.g. RIP, or the sufficient conditions for recovery with Lasso that we discussed above)
do not hold anymore. For instance, in gene array data \cite{zou2005regularization}, it is known that many columns of the design matrix are correlated, and that RIP does not hold. It is therefore crucial to come up with algorithms for which recovery is provably possible under different conditions, which we tackle in this paper.

\paragraph{$k$-support Norm Regularization.} We now introduce the $k$-support norm, which is the main component of our algorithm, as well as its proximal operator. The $k$-support norm was first introduced in \cite{argyriou2012sparse}, as the tightest convex relaxation of the intersection of the $\ell_2$ ball and the $\ell_0$ ball.
It was later generalized to the matrix case \cite{mcdonald2016fitting,mcdonald2016new}, as well as successfully applied to several problems, including for instance fMRI \cite{gkirtzou2013fmri,belilovsky15b}. We give below its formal definition, with the following variational formula from \cite{argyriou2012sparse}:

\begin{definition}[\cite{argyriou2012sparse,mcdonald2014spectral}]
  Let $k \in \{1, ..., d\}$. The $k$-support norm $ \| \cdot \|_k^{sp}$ is defined, for every $\w \in \R^d$, as:
  \begin{align*}
    \| \w \|_k^{sp} = \min & \left\{\sum_{I \in \mathcal{G}_k} \left\|\bm{v}_I\right\|_2 : \bm{v}_I \in \mathbb{R}^d, \operatorname{supp}\left(\bm{v}_I\right) \subseteq I, \right. \\
    &\left. ~~\sum_{I \in \mathcal{G}_k} \bm{v}_I = \bm{w} \right\}
  \end{align*}
            where $\mathcal{G}_k$ denotes the set of all subsets of $\{1, . . . , d\}$ of cardinality at most $k$.
  \end{definition}
In other words, the $k$-support norm is equal to the smallest sum of the norms of some $k$-sparse \textit{atoms} (the $\y_I$ above) that constitute $\w$: as studied in \cite{chatterjee2014}, the $k$-support norm is indeed a so-called \textit{atomic norm}. One can also see from this definition that the $k$-support norm interpolates between the $\ell_1$ norm (which it is equal to if $k=1$) and the $\ell_2$ norm (which it is equal to if $k=d$). As discussed in \cite{argyriou2012sparse}, another interpretation of the $k$-support norm is that it is equivalent to the Group-Lasso penalty with overlaps \cite{jacob2009group}, when the set of overlapping groups is all possible subsets of $\{1, . . . , d\}$ of cardinality at most $k$.
Finally, we introduce the proximal operator  \cite{parikh2014proximal} below, that will be used in our algorithm:
\begin{definition}[Proximal operator, \cite{parikh2014proximal}]
  The proximal operator for a function $h: \R^d \rightarrow \R$ is defined as:
  $$ \prox_{h}(\z) = \arg\min_{\w} h(\w) + \frac{1}{2}\|\w - \z\|_2^2$$
\end{definition}
A closed form for the proximal operator of the squared $k$-support norm was first given in \cite{argyriou2012sparse}, and more efficient computations have been found e.g. in \cite{mcdonald2016new}, which we will use in IRKSN, as described in Appendix~\ref{sec:proxksn}. 
\section*{The Algorithm}
In this section, we describe the IRKSN (Iterative Regularization with $k$-Support Norm) algorithm. It is based on the general accelerated algorithm from \cite{matet2017}, in which we plug a regularization function based on the $k$-support norm. More precisely, \cite{matet2017} describe a general regularization algorithm for model recovery based on a primal-dual method, and an early stopping rule. As they do, we will solve the following problem approximately (i.e. with early stopping):
\begin{align*}
 (I_{ks}): \quad \quad &\min_{\w} R(\w) \nonumber\\
\text{s.t.} \quad & \XX \w = \y^{\delta} 
\end{align*}
with a specific regularizer that we introduce: $R(\w) =  F(\w)  + \frac{\alpha}{2} {\|\w \|_2}^2$ with $ F(\w) =  \frac{1 - \alpha}{2}(\|\w \|^{sp}_{k})^2 $, for some constant $1 > \alpha > 0 $ which will be described later.
The algorithm that we will use to solve approximately $(I_{ks})$ is the Accelerated Dual Gradient Descent (ADGD) described in \cite{matet2017}, which is an accelerated version of a primal-dual method that is known in the literature under many names, and that comprises the following steps, with $\gamma$ being some learning rate, and $\hat{\bm{v}}_t$ being a dual variable:\\
$\texttt{\#  primal projection step}$ \\
$\hat{\w}_t \leftarrow \text{prox}_{\alpha^{-1} F}(-\alpha^{-1}\XX^{\top} \hat{\bm{v}}_{t})$\\
$\texttt{\# dual update step}$\\
$ \hat{\bm{v}}_{t+1} \leftarrow \hat{\bm{v}}_{t} + \gamma (\XX \hat{\w}_t - \bm{y}^{\delta})$\\
The method above is most commonly known in the signal processing and image denoising literature as Linearized Bregman Iterations, or Inverse Scale Space Methods \cite{cai2009linearized,osher2016}. In the optimization literature, it is mostly known as (Lazy) Mirror Descent \cite{bubeck2015convex}, also called Dual Averaging \cite{nesterov2009,Xiao09}. The main idea in \cite{matet2017} is to early stop the algorithm at some iteration $T$, before convergence. We present the full accelerated version, IRKSN, in Algorithm \ref{alg:irksn}. 
\begin{algorithm}[!hb]
  \caption{IRKSN}
  \label{alg:irksn}
  \begin{algorithmic}
    \STATE {\bfseries Input:} $\bm{\hat{v}}_0=\bm{\hat{z}}_{-1}=\bm{\hat{z}}_0 \in \mathbb{R}^d, \gamma=\alpha\|\XX\|^{-2},
    \theta_0=1$ 
    \FOR {$t=0$ {\bfseries to} $T$}
    \STATE $\hat{\w}_t\leftarrow \operatorname{prox}_{\alpha^{-1} F}\left(-\alpha^{-1} \XX^T \bm{\hat{z_t}}\right)$ \\
    \STATE $\bm{\hat{r}}_t\leftarrow \operatorname{prox}_{\alpha^{-1} F}\left(-\alpha^{-1} \XX^T \bm{\hat{v}}_t\right)$ \\
    \STATE     $ \bm{\hat{z}}_t\leftarrow \bm{\hat{v}}_t+\gamma\left(\XX \bm{\hat{r_t}}-\bm{y}^{\delta}\right)$\\
    \STATE $\theta_{t+1}\leftarrow \left(1+\sqrt{1+4 \theta_t^2}\right) / 2$\\
    \STATE $\bm{\hat{v}}_{t+1}=\bm{\hat{z}}_t+\frac{\theta_t-1}{\theta_{t+1}}\left(\bm{\hat{z}}_t-\bm{\hat{z}}_{t-1}\right)$\\
    \ENDFOR 
  \end{algorithmic}
\end{algorithm}
\section*{Main Results}
In this section, we introduce the main result of our paper, which gives specific conditions for robust recovery of $\w^*$, and early stopping bounds on $\|\hat{\w}_t - \w^*\|$ for IRKSN.

\subsection*{Assumptions}

We will present several sufficient conditions for recovery with the $k$-support norm, which are similar to the sufficient conditions needed for $\ell_1$-based recovery that we describe in Assumption~\ref{ass:assl1} (we will then elaborate on the differences between such conditions).
The first assumption below is a variant of the usual feasibility assumption of the noiseless problem \cite{foucart2013}: it simply states that $\w^*$, the true model that we wish to recover, is a feasible solution of the noiseless problem,  and that it is $k$-sparse. Additionally, if several feasible solutions of same support than $\w^*$ exist, $\w^*$ should be the smallest norm one (we will elaborate on such condition in this section). Recall from the Introduction that $\bm{y}$ is the true target vector, i.e. uncorrupted by noise.
\begin{assumption}\label{ass:sol}
  $\w^*$ is $k$-sparse of support $S \subset [d]$, and is a solution of the system $(L): \XX \w = \y$. In addition, $\w^*$  is the smallest $\ell_2$ norm solution of $(L)$ on its support, that is,  $\w^*$ is such that:  
  $$\w_S^* = \arg \min_{\z \in \R^{k}: \XX_S \z = \y} \| \z \|_2$$
\end{assumption}
We now provide our main assumption, which is intrinsically linked to the structure of the $k$-support norm, and which is, up to our knowledge, the first condition of such kind in the sparse recovery literature.
\begin{assumption}\label{ass:ass}  $\w^*$ verifies:
  $$\max_{\ell \in \bar{S}}| \langle \XX_S^{\dagger} \bmx_{\ell}, \w^*_{S}\rangle| < \min_{j \in S}| \langle \XX_S^{\dagger} \bmx_{j}, \w^*_{S}\rangle| $$
\end{assumption}
Up to our knowledge, we are the first to provide such assumptions for recovery with a $k$-support norm based algorithm: although \cite{chatterjee2014} proposed a $k$-support norm based algorithm and corresponding conditions for recovery, those conditions only apply in the case of a design matrix $\bm{X}$ with values which are i.i.d. samples from a Gaussian distribution. 

\subsection*{Discussion on the Assumptions}

In this section, we attempt to interpret the assumptions above in simple terms, and to compare them to some similar sufficient conditions for recovery with $\ell_1$ norm. More precisely, the condition below implies Condition 4.3 from \cite{grasmair2011}, which latter is shown in \cite{grasmair2011} to be a necessary and sufficient condition for achieving a linear rate of recovery with $\ell_1$ norm Tikhonov regularization. We prove such  implication in Appendix \ref{app:recall}.

\begin{assumption}[Recovery with $\ell_1$ norm.]\label{ass:assl1}
   Let $\w^*$ be supported on a support $S \subset [d]$. $\w^*$ is such that:
   \begin{enumerate}[label={(\roman*)},align=left]
  \item $\XX\w^* = \bm{y}$
  \item $\XX_{S}$ is injective
  \item $\max_{\ell \in \bar{S}}|\langle\XX_S^{\dagger} \bmx_{\ell}, \operatorname{sgn}(\w_S^*)\rangle|<1 $
  \end{enumerate}  
\end{assumption}

Below, we now compare this assumption to ours.

\textbf{The min $\ell_2$ norm solution.}
  In our Assumption~\ref{ass:sol}, the minimum $\ell_2$ norm condition is actually not restrictive, compared to Assumption~\ref{ass:assl1}: indeed, in Assumption~\ref{ass:assl1}  $\XX_S$ needs to be injective, which implies that there needs to be  only one solution $\w^*_S$ on $S$ such that $\XX_S\w_s^* = \bm{y}$: we can also work in such situations, but we also include  the additional cases where there are several solutions on $S$ (we just require that $\w^*$ is the minimum norm one) : $\XX_S$ does not need to be injective in our case. Importantly we can deal with cases with $n < k$, when Lasso (and $\ell_1$ iterative regularization methods) cannot (that is, we can obtain recovery in a regime where the number of samples $n$ is even lower than the sparsity of the signal $k$).  Note that for the Lasso, the condition $n \geq k$ is even \textit{necessary}: indeed, when $n < k$, the Lasso is known to saturate \cite{zou2005regularization} and recovery is impossible: interestingly, there is no such constraint when using a $k$-support norm regularizer (similarly to recovery with ElasticNet).
  
\textbf{Dependence on the sign.}
  As we can observe, Assumption~\ref{ass:assl1} is verified or not based on $\operatorname{sgn}(\w_S^*)$. This implies that irrespective of the actual values of $\w^*$, recovery will be possible or not only based on $\operatorname{sgn}(\w_S^*)$. On the contrary, our Assumption \ref{ass:ass} depends on $\w^*$ itself.

\textbf{Case where $\XX_S$ is injective.} In the case where $\XX_S$ is injective (as will happen in most cases in practice when $n > k$, i.e. unless there is some spurious exact linear dependence between columns), it is even easier to compare Assumptions \ref{ass:ass} and \ref{ass:assl1}. Indeed, since in that case we have that $\XX_S$ is full column rank,  we then have : $\XX_S^{\dagger} \XX_S = \bm{I}_{k \times k}$. 
Therefore, Assumption \ref{ass:ass} can be rewritten into: $\max_{\ell \in \bar{S}}| \langle \XX_S^{\dagger} \bmx_{\ell}, \w^*_{S}\rangle| < \min_{j \in S}| w^*_i|$, which is equivalent to:
$$\max_{\ell \in \bar{S}}| \langle \XX_S^{\dagger} \bmx_{\ell}, \frac{\w^*_{S}}{\min_{j \in S}| w^*_i|}\rangle| < 1 $$
Therefore, we can notice that if $\w^*_S = \gamma \operatorname{sgn}(\w_S^*)$ for some $\gamma >0$  (that is, each component of $\w^*_S$ have the same absolute value), both Assumptions \ref{ass:ass} and \ref{ass:assl1} become equivalent (because then: $\frac{\w^*_{S}}{\min_{j \in S}| w^*_i|} = \operatorname{sgn}(\w_S^*)$). However, the two conditions \ref{ass:ass} and \ref{ass:assl1} may differ depending on the \textit{relative magnitudes} of the entries in $\w^*_S$. In particular, it may happen that our Assumption~\ref{ass:ass} is verified even if the Assumption~\ref{ass:assl1} is not verified. We analyze such an example in Example 1.
\subsection*{Early Stopping Bound}
We are now ready to state our main result: 
\begin{theorem}[Early Stopping Bound]\label{thm:thm} Let $\delta\in\left]0,1\right]$ and let  $(\hat{\w}_t)_{t\in\mathbb{N}}$  be  the sequence generated by IRKSN.
  Assuming the design matrix $\XX$ and the true sparse vector $\w^*$ satisfy Assumptions~\ref{ass:sol}  and \ref{ass:ass}, and with $\alpha <  \frac{\eta }{\|\w\|_{\infty}}$ with $\eta := \min_{j \in S}| \langle (\XX_{S} \XX_{S}^{\top})^{\dagger} \bm{y}, \bmx_{j} \rangle|- \max_{\ell \in \bar{S}}| \langle  (\XX_{S} \XX_{S}^{\top})^{\dagger} \bm{y}, \bmx_{\ell} \rangle|$, we have for $t \geq 2$: 
  $$ \| \hat{\w}_t - \w^*\|_2 \leq a t \delta + b t^{-1}$$ $$\text{with} \quad a = 4 \|\XX\|^{-1} \quad \text{and}\quad b = \frac{2\| \XX\| \| (\XX_S^{\top})^{\dagger} \w_S^*\|}{\alpha}$$
  In particular (if $\delta > 0$), with $t_{\delta} = \lceil c \delta^{-1/2} \rceil$, for some $c > 0$:
    $$ \| \hat{\w}_t - \w^*\|_2 \leq (a(c+1) + b c^{-1}) \delta^{1/2}$$
  \end{theorem}
  \begin{proof}
Proof in Appendix \ref{proof:thm}.    
  \end{proof}
\paragraph{Discussion.}
We can notice in Theorem \ref{thm:thm} above that  $b$ is large when $\alpha$ is small: therefore, if the inequality in \ref{ass:ass} is very tight, as a consequence, $\alpha$ will need to be taken small, and $b$ will become large. Therefore, we can say that the larger the margin by which Assumption \ref{ass:ass} is fulfilled is, the better the retrieval of the true vector $\w^*$ is (because the larger we can choose $\alpha$).
\section*{Illustrating Example} 
In this section, we describe a simple example that illustrates the cases where $\ell_1$ norm-based regularization fails, and where IRKSN will successfully recover the true vector. 

\textbf{Example 1. } We consider a model that consists of three ``generating'' variables $X^{(0)}, X^{(1)}$ and $X^{(2)}$, that are random i.i.d. variables from standard Gaussian  (we denote $X^{(0)} \sim \mathcal{N}(0, 1)$ and $X^{(1)} \sim \mathcal{N}(0, 1)$ and $X^{(2)} \sim \mathcal{N}(0, 1)$). 
Two other variables $X^{(3)}$ and $X^{(4)}$, are actually correlated with the previous random variables: they are obtained noiselessly, and linearly from those, with some vectors $ \bm{w}^{(3)}$ and $ \bm{w}^{(4)}$ that will be defined below:
$$X^{(3)} =  w^{(3)}_0 X^{(0)} + w^{(3)}_1 X^{(1)} + w^{(3)}_2 X^{(2)}$$ and $$X^{(4)} =  w^{(4)}_0 X^{(0)} + w^{(4)}_1 X^{(1)} + w^{(4)}_2 X^{(2)}$$
In addition, similarly, the actual observations $Y$ are formed noiselessly and linearly from $( X^{(0)},  X^{(1)},  X^{(2)})$, for some vector $\bm{w}^{(y)}$:
$$Y =  w^{(y)}_0 X^{(0)} + w^{(y)}_1 X^{(1)} + w^{(y)}_2 X^{(2)}$$
A graphical visualization of this construction can be seen on Figure \ref{fig:toyxpgraph}. More precisely, we define the vectors $\bm{w}^{(3)}, \bm{w}^{(4)}$ and $\bm{w}^{(y)}$ are defined as follows:
\begin{equation*}
      \bm{w}^{(3)} = \begin{bmatrix}
               9/11\\
                6/11\\
                2/11\\
                0\\
                0\\   
 \end{bmatrix},
   \bm{w}^{(4)} = \begin{bmatrix}
               1/3\\
                14/15\\
                2/15\\
                0\\
                0\\   
                                      \end{bmatrix},
        \bm{w}^{(y)} = \begin{bmatrix}
               1\\
                1\\
                -4\\
                0\\
                0\\   
	\end{bmatrix}.
	\end{equation*}
    We will generate such a dataset with $n=4$: so the dataset will be composed of 4 samples of $X^{(0)}, X^{(1)}, X^{(2)}, X^{(3)}, X^{(4)}$, which form the matrix $\XX \in \R^{4, 5}$, with $\XX = \left[\bmx_0, \bmx_1, \bmx_2, \bmx_3 , \bmx_4 \right]$  and 4 samples of $Y$, which form the vector $\bm{y} \in \R^4$. In our case, we have $S= \text{supp}(\w^{(y)}) = \{0, 1, 2 \}$, and therefore we just ensure that $\XX_S = \left[\bmx_0, \bmx_1, \bmx_2 \right]$ is full column rank  (which should be the case with overwhelming probability since those three first vectors are sampled from a Gaussian, and since we have $n=4 > k=3$). Our goal is to reconstruct the true linear model of $Y$, which is $\bm{w}^{(y)}$ from the observation of $\XX$ and $\y$. 
\begin{figure}[htp]
  \begin{center}
  \begin{tikzpicture}[-, scale=.6, line width=0.5pt, inner sep=0.2mm, shorten >=.1pt, shorten <=.1pt]
\draw (0, 0) node(1) [circle, draw] {{\footnotesize\,$X^{(0)}$\,}};
  \draw (1, -1.5) node(2) [circle, draw] {{\footnotesize\,$X^{(1)}$\,}};
\draw (3, -2) node(3) [circle, draw] {{\footnotesize\,$X^{(2)}$\,}};
\draw (5, -1.5) node(4) [circle, draw] {{\footnotesize\,$X^{(3)}$\,}};
\draw (6, 0) node(5) [circle, draw] {{\footnotesize\,$X^{(4)}$\,}};
\draw (3, 2) node(6) [circle, draw] {{\footnotesize\,$~Y^{~}$\,}};
  \draw[-arcsq, thick] (1) --(6); 
    \draw[-arcsq, thick] (2) --(6); 
    \draw[-arcsq, thick] (3) edge (6); 
      \draw[-arcsq, thick] (1) --(4); 
    \draw[-arcsq, thick] (2) --(4); 
    \draw[-arcsq, thick] (3) edge (4); 
          \draw[-arcsq, thick] (1) --(5); 
    \draw[-arcsq, thick] (2) --(5); 
  \draw[-arcsq, thick] (3) edge (5); 
        \end{tikzpicture}
\caption{$X^{(3)}$, $X^{(4)}$ are correlated with $X^{(0)}, X^{(1)}$, $X^{(2)}$}
\label{fig:toyxpgraph}
  \end{center}
\end{figure}
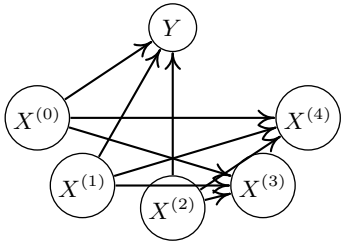
    We can easily check mathematically (using the closed form from the first column of Table \ref{tab:pen}), that this example only verifies our condition (Assumption \ref{ass:ass}), but that it does not verify Assumption \ref{ass:assl1} (i.e. it is in the blue area from Figure \ref{fig:imp}). Indeed, in that case, $\XX_S$ is full column rank, which implies $(\XX_S)^{\dagger}\bmx_3 = \bm{w}^{(3)}$ and $(\XX_S)^{\dagger}\bmx_4 =\bm{w}^{(4)}$ \cite{golub2013matrix}.
    We then have:
    \begin{align*}
      |\langle\XX_S^{\dagger} \bmx_3, \operatorname{sgn}(\w^{(y)})\rangle| &= | \langle \bm{w}^{(3)}, \operatorname{sgn}(\w^{(y)})\rangle | = 13/11>1      
    \end{align*}
    \begin{align*}
     |\langle\XX_S^{\dagger} \bmx_4, \operatorname{sgn}(\w^{(y)})\rangle| &= | \langle \bm{w}^{(4)}, \operatorname{sgn}(\w^{(y)})\rangle |= 17/15>1      
    \end{align*}
    Therefore: $\max_{\ell \in \bar{S}}|\langle\XX_S^{\dagger} \bmx_\ell, \operatorname{sgn}(\w_S^*)\rangle| = \frac{13}{11} >1$
    Which means that Assumption \ref{ass:assl1} is not verified.    However, on the other hand, we have:
    \begin{align*}
      | \langle\XX_S^{\dagger} \bmx_3, \frac{\bm{w}^{(y)}}{\min_{j \in S}| w^{(y)}_i|}\rangle| &= | \langle \bm{w}^{(3)},\frac{\bm{w}^{(y)}}{\min_{j \in S}| w^{(y)}_i|} \rangle | = \frac{7}{11}    
    \end{align*}
    \begin{align*}
      |\langle\XX_S^{\dagger} \bmx_4, \frac{\bm{w}^{(y)}}{\min_{j \in S} | w^{(y)}_i|}\rangle | &= | \langle \bm{w}^{(4)}, \frac{\bm{w}^{(y)}}{\min_{j \in S}| w^{(y)}_i|}\rangle | = \frac{11}{15}     
    \end{align*}
    Therefore: $    \max_{\ell \in \bar{S}}|\langle\XX_S^{\dagger} \bmx_\ell, \frac{\bm{w}^{(y)}}{\min_{j \in S}| w^{(y)}_i|} \rangle | = \frac{11}{15} <1$.
    Therefore, from the Section \textit{Discussion on the Assumptions}, paragraph \textit{Case where $\XX_S$ is injective}, we see that our Assumption \ref{ass:ass} is verified here.
      \paragraph{Comparison of the IRKSN path with Lasso.} 
  In Figure \ref{fig:path} below, we compare the Lasso path (that is, the solutions found by Lasso for all values of the penalization $\lambda$), with the IRKSN path (that is, the solutions found by IRKSN at every timestep).  For indicative purposes, we also provide the path of the ElasticNet on the same problem in Appendix \ref{app:path}.
      \begin{figure}[!bht]
  \centering
  \subfigure[Lasso path]{\includegraphics[scale=0.23]{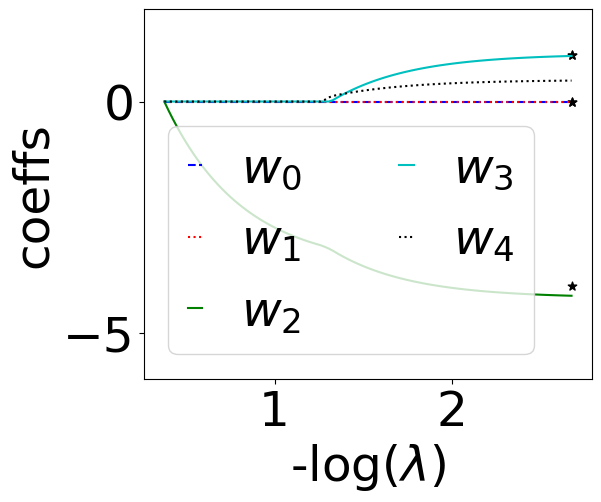}\label{fig:lassopath}} 
  \subfigure[IRKSN path]{\includegraphics[scale=0.23]{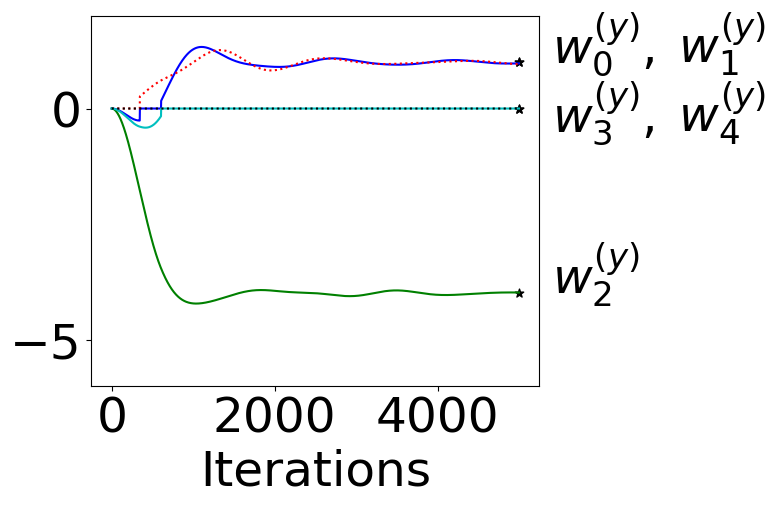}\label{fig:irksnpath}}
   \caption{Comparison of the path of IRKSN with Lasso. $w^{(y)}_i$
   is the $i$-th component of $\bm{w}^{(y)}$, and $\lambda$ is the penalty of the Lasso. We recall $w^{(y)}_0=w^{(y)}_1=1, w^{(y)}_2=-4, w^{(y)}_3=w^{(y)}_4=0$: only IRKSN recovers the true $\bm{w}^{(y)}$.}
   \label{fig:path}
 \end{figure}

 \begin{figure}[!bht]
  \centering
  \subfigure[Model error $ \|\hat{\bm{w}} - \bm{w}^{(y)} \|$]{\includegraphics[scale=0.25]{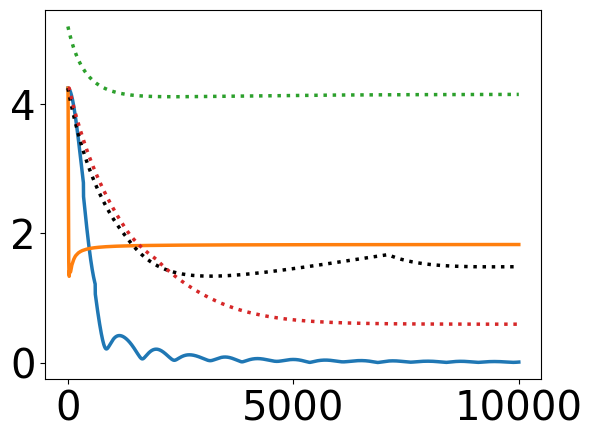}\label{fig:toycp}}\quad 
    \subfigure[Model sparsity $\|\hat{\bm{w}}\|_0$]{\includegraphics[scale=0.25]{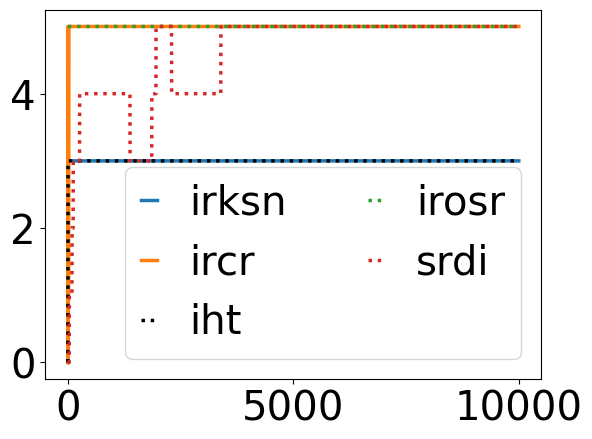}\label{fig:spa}}
   \caption{Error and sparsity vs. number of iterations. Only IRKSN can recover the true $\bm{w}^{(y)}$ in this example.}
   \label{fig:bothfig}
 \end{figure}

 As we can see, the Lasso is unable to retrieve the true sparse vector, for any $\lambda$. However IRKSN can successfully retrieve it, which confirms the theory above.

 In addition, this path from Figure \ref{fig:path} above illustrates well the optimization dynamics of IRKSN: first, the true support of $\w^{(y)}$ is not identified in the first iterations. But after a few iterations, we observe what we could call a phenomenon of \textit{exchange of variable}:  $w^{(y)}_0$ is exchanged with $w^{(y)}_1$, and later, $w^{(y)}_3$ is exchanged with $w^{(y)}_0$ (by \textit{exchange}, we mean that at a timestep $t$, $w^{(y)}_0(t) \neq 0$ but $w^{(y)}_1(t)=0$, but at timestep $t + 1$: $w^{(y)}_0(t+1) = 0$ and $w^{(y)}_1(t+1)\approx w^{(y)}_0(t)$). This can be explained by the fact that when $\alpha$ is small, the proximal operator of the $k$-support norm approaches the hard-thresholding operator from \cite{blumensath2009}: hence at a particular timestep the ordering (in absolute magnitude) of the components of $\XX^{\top} \hat{\bm{z}}_t$ suddenly changes (with the components where the change occurs having about the same magnitude at the time of change, if the learning rate is small), which results into such an observed change in primal space. Additionally, in Figure \ref{fig:bothfig}, we run the iterative methods from Table \ref{tab:pen} (IRKSN, IRCR, IROSR and SRDI) (as well as IHT for comparison) on Example 1, and measure the recovery error $\| \hat{\bm{\w}} - \bm{w}^{(y)} \|$ as well as the sparsity $\| \hat{\bm{\w}} \|_0$ of the iterates. As we can see, only IRKSN can achieve 0 error, that is, full recovery in the noiseless setting. In addition, except IHT (which however fails to approach the true solution), no method is able to converge to a 3-sparse solution, which is the true degree of sparsity of the solution.

 \section*{Experiments}
 Below we present experimental results to evaluate the sparse recovery properties of IRKSN. Additional details on those experiments as well as further experiments are provided in the Appendix.

\paragraph*{Experimental Setting.}


We consider a simple linear regression setting with a correlated design matrix,  i.e. where the design matrix $\X$ is formed by $n$ i.i.d. samples from $d$ (we take $d=50$ here) correlated Gaussian random variables $\{X_1, .., X_d\}$ of zero mean and unit variance,  such that: $\forall i \in \{1, \ldots, d\}: \mathbb{E}[X_i] = 0, \mathbb{E}[X_i^2] = 1;$ and $\forall (i, j) \in \{1, \ldots, d\}^2, i\neq j: \mathbb{E}[X_i X_j] = \rho^{|i-j|}$.
More precisely, we generate each feature $X_i$ in an auto-regressive manner, from previous features, using a correlation $\rho \in [0, 1)$, in the following way: we have $X_1 \sim \mathcal{N}(0, 1)$ and $\sigma^2 = 1 - \rho^2$, and for all $j \in \{2, ..., d\}$: $X_{j+1} = \rho X_j + \epsilon_j$ where $\epsilon_j = \sigma * \Delta$, with $\Delta \sim \mathcal{N}(0, 1)$. Additionally, $\w$ is supported on a support, sampled uniformly at random, of $k=10$ non-zero entries, with each non-zero entry sampled from a normal distribution, and $\y$ is obtained with a noise vector $\bm{\epsilon}$ created from i.i.d. samples from a normal distribution, rescaled to enforce a given signal to noise ratio (SNR), as follows:  $\y = \X \w^* + \bm{\epsilon}$
with the signal to noise ratio defined as $\text{snr} = \frac{\| \X \w^*\|}{\|\bm{\epsilon}\|}$. We generate this dataset using the \texttt{make\_correlated\_data} function from the \texttt{benchopt} package \cite{moreau2022benchopt}. 
Such a dataset is commonly used to evaluate sparse recovery algorithms (see e.g. \cite{molinari2021iterative}), since it possesses correlated features, which is more challenging for sparse recovery (see e.g. the ElasticNet paper, which was motivated by such correlated datasets \cite{zou2005regularization}). In addition, the advantage of such synthetic dataset is that the support is known since it is generated, which therefore allows to evaluate the performance of the algorithms on support recovery, contrary to real-life datasets where a true sparse support of $\w$ is hypothetical (or at least often unknown). Additionally, we can notice that such dataset resembles our Example 1, as some features are generated from other features. We evaluate the performance of each final recovered model $\w$ using the F1 score on support recovery,  defined as follows: $\text{F1} = 2 \frac{P R}{P + R}$, with $P$ the precision and $R$ the recall of support recovery, which are defined as: $P = \frac{|\text{supp}(\w^*) \cap \text{supp}(\w) |}{| \text{supp}(\w)|}$ and
$R = \frac{|\text{supp}(\w^*) \cap \text{supp}(\w) |}{| \text{supp}(\w^*)|}$. Therefore, the F1 score allows to evaluate at the same time how much of the predicted nonzero elements are accurate, and how much of the actual support has been found. A higher F1 score indicates better identification of the true support. In each experiment (defined by a particular value of $n, \rho, \text{snr}$ and a given random seed for generating $X$, $\w^*$ and $\bm{\epsilon}$), and for each algorithm, we choose the hyperparameters from a grid-search, to attain the best F1-score (we give details on that grid in the Appendix). For all algorithms which need to set a value $k$ (IRKSN, KSN, IHT), we set $k$ to its true value $k=10$. In a realistic use-case, since the support is unknown, one may instead tune those hyper-parameters based on a hold-out validation set prediction mean squared error, but tuning those hyperparameters directly for best support F1 score, as we do, allows to evaluate the best potential support recovery capability of each algorithm (e.g. for Lasso it informs us that \textit{there exist} a certain $\lambda$, such that we can achieve such a support recovery score).
Each experiment is regenerated 5 times with different random seeds, and the average of the obtained best F1 scores, as well as their standard deviation, are reported in Figures \ref{fig:n}, \ref{fig:rho}, and \ref{fig:snr}, for various values of the dataset parameters, while the others are kept fixed.
\begin{figure}[!bht]
  \centering
  \subfigure[F1-score vs. $n$]{\includegraphics[scale=0.27]{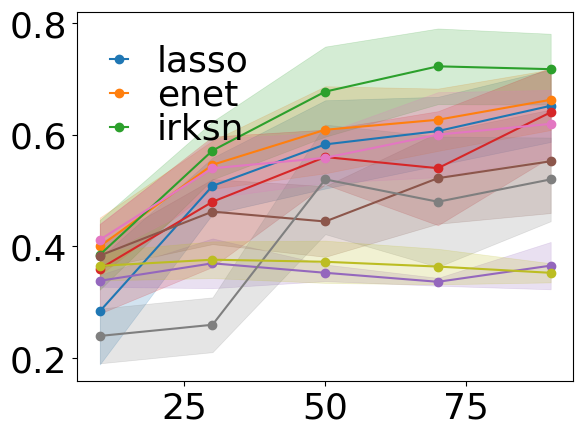}\label{fig:n}}\quad 
    \subfigure[F1-score vs. snr]{\includegraphics[scale=0.27]{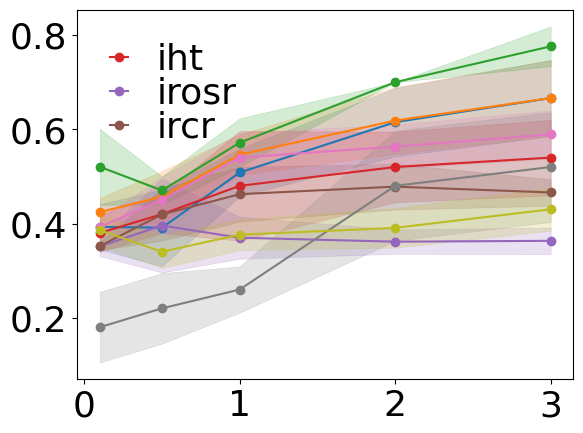}\label{fig:snr}}
        \subfigure[F1-score vs. $\rho$]{\includegraphics[scale=0.27]{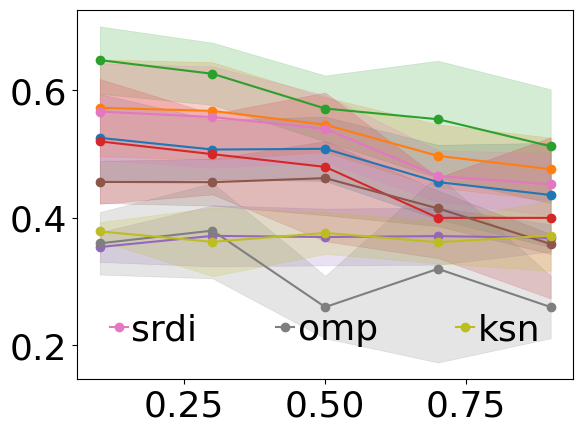}\label{fig:rho}}
                \subfigure[F1-score vs. t]{\includegraphics[scale=0.25]{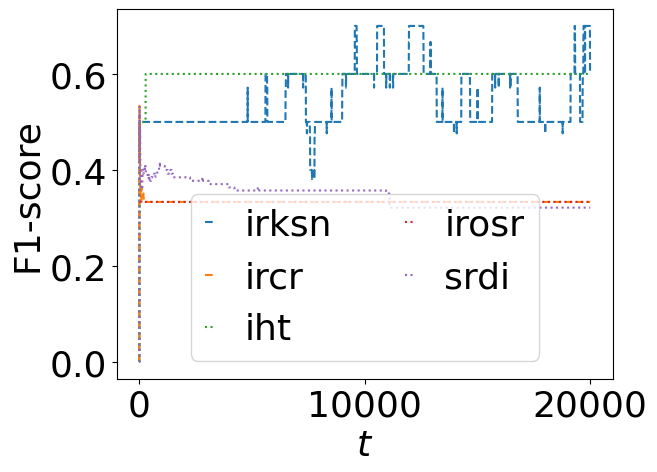}\label{fig:f1iter}}
   \caption{F1-score of support recovery in various settings}
   \label{fig:f1score}
 \end{figure}
In Figure \ref{fig:n}, we take $\rho=0.5$, $\text{snr}=1.$, and $n \in \{10, 30, 50, 70, 90 \}$. In Figure \ref{fig:snr}, we take $\rho=0.5$, $\text{snr} \in \{0.1, 0.5, 1., 2., 3.\}$, and $n=30$. In Figure \ref{fig:rho}, we take $\rho \in \{0.1, 0.3, 0.5, 0.7, 0.9 \}$, $\text{snr}=1.$, and $n=30$.  Additionally, we plot on Figure \ref{fig:f1iter} the evolution of the F1 score along training for iterative algorithms (i.e. algorithms where there is no grid search over a penalty $\lambda$, which are IHT, IRKSN, IRCR, IROSR, SRDI), in the case where $n=30, \text{snr}=3$, and $\rho=0.5$.

\paragraph*{Results.}

In all the experiments, as can be expected, we observe that support recovery is more successful when the signal to noise ratio is high, the number of samples is greater, and the correlation $\rho$ is smaller (for that latter point, this is due to the fact that highly correlated datasets are harder for sparse recovery, see e.g. \cite{zou2005regularization} for a discussion on the topic). But overall, we can observe that IRKSN consistently achieves better support recovery than other algorithms from Table \ref{tab:pen}. Also, we can observe on Figure \ref{fig:f1iter} that IHT and IRKSN maintain a good F1 score after many iterations, while other methods implicitly enforcing an $\ell_1$ norm regularization (IRCR, IROSR, SRDI) have poor F1 score in late training.

\section*{Conclusion}
In this paper, we introduced an iterative regularization method based on the $k$-support norm regularization, IRKSN, to complement usual methods based on the $\ell_1$ norm.
In particular, we gave some condition for sparse recovery with our method, that we analyzed in details and compared to traditional conditions for recovery with $\ell_1$ norm regularizers, through an illustrative example. We then gave an early stopping bound for sparse recovery with IRKSN with explicit constants in terms of the design matrix and the true sparse vector. Finally, we evaluated the applicability of IRKSN on several experiments. In future works, it would be interesting to analyze recovery with the $s$-support norm for general $s$, where $s$ is not necessarily equal to $k$: indeed, this setting would generalize both our work and works based on the $\ell_1$ norm. We leave this for future work.
\section*{Acknowledgements}
We would like to thank Velibor Bojkovi\'c for fruiteful discussions, as well as the anonymous reviewers for their useful comments. Xiao-Tong Yuan is funded in part by the National Key Research and Development Program of China under Grant No. 2018AAA0100400, and in part by the Natural Science Foundation of China (NSFC) under Grant No.U21B2049 and No.61936005.

\bigskip
\noindent 

\bibliography{aaai24}

\begin{thebibliography}{55}
\providecommand{\natexlab}[1]{#1}

\bibitem[{Abraham et~al.(2014)Abraham, Pedregosa, Eickenberg, Gervais, Mueller,
  Kossaifi, Gramfort, Thirion, and Varoquaux}]{Chamma_nilearn}
Abraham, A.; Pedregosa, F.; Eickenberg, M.; Gervais, P.; Mueller, A.; Kossaifi,
  J.; Gramfort, A.; Thirion, B.; and Varoquaux, G. 2014.
\newblock Machine learning for neuroimaging with scikit-learn.
\newblock \emph{Frontiers in neuroinformatics}, 8: 14.

\bibitem[{Argyriou, Foygel, and Srebro(2012)}]{argyriou2012sparse}
Argyriou, A.; Foygel, R.; and Srebro, N. 2012.
\newblock Sparse prediction with the $ k $-support norm.
\newblock \emph{Advances in Neural Information Processing Systems}, 25.

\bibitem[{Baptiste, Urruty, and Lemarechal(2001)}]{Baptiste01}
Baptiste, J.; Urruty, H.; and Lemarechal, C. 2001.
\newblock Fundamentals of convex analysis.

\bibitem[{Bauschke, Combettes et~al.(2011)}]{bauschke2011convex}
Bauschke, H.~H.; Combettes, P.~L.; et~al. 2011.
\newblock \emph{Convex analysis and monotone operator theory in Hilbert
  spaces}, volume 408.
\newblock Springer.

\bibitem[{Belilovsky et~al.(2015)Belilovsky, Gkirtzou, Misyrlis, Konova,
  Honorio, Alia-Klein, Goldstein, Samaras, and Blaschko}]{belilovsky15b}
Belilovsky, E.; Gkirtzou, K.; Misyrlis, M.; Konova, A.~B.; Honorio, J.;
  Alia-Klein, N.; Goldstein, R.~Z.; Samaras, D.; and Blaschko, M.~B. 2015.
\newblock Predictive sparse modeling of fMRI data for improved classification,
  regression, and visualization using the k-support norm.
\newblock \emph{Computerized Medical Imaging and Graphics}, 46: 40--46.

\bibitem[{Ben-Israel and Greville(2003)}]{ben2003generalized}
Ben-Israel, A.; and Greville, T.~N. 2003.
\newblock \emph{Generalized inverses: theory and applications}, volume~15.
\newblock Springer Science \& Business Media.

\bibitem[{Bertrand and Massias(2021)}]{bertrand2021anderson}
Bertrand, Q.; and Massias, M. 2021.
\newblock Anderson acceleration of coordinate descent.
\newblock In \emph{International Conference on Artificial Intelligence and
  Statistics}, 1288--1296. PMLR.

\bibitem[{Blumensath and Davies(2009)}]{blumensath2009}
Blumensath, T.; and Davies, M.~E. 2009.
\newblock Iterative hard thresholding for compressed sensing.
\newblock \emph{Applied and computational harmonic analysis}, 27(3): 265--274.

\bibitem[{Boyd, Boyd, and Vandenberghe(2004)}]{boyd2004convex}
Boyd, S.; Boyd, S.~P.; and Vandenberghe, L. 2004.
\newblock \emph{Convex optimization}.
\newblock Cambridge university press.

\bibitem[{Breheny(2022)}]{breheny2022}
Breheny, P. 2022.

\bibitem[{Bubeck et~al.(2015)}]{bubeck2015convex}
Bubeck, S.; et~al. 2015.
\newblock Convex optimization: Algorithms and complexity.
\newblock \emph{Foundations and Trends{\textregistered} in Machine Learning},
  8(3-4): 231--357.

\bibitem[{Cai, Osher, and Shen(2009)}]{cai2009linearized}
Cai, J.-F.; Osher, S.; and Shen, Z. 2009.
\newblock Linearized Bregman iterations for compressed sensing.
\newblock \emph{Mathematics of computation}, 78(267): 1515--1536.

\bibitem[{Chang and Lin(2011)}]{libsvm}
Chang, C.-C.; and Lin, C.-J. 2011.
\newblock {LIBSVM}: A Library for Support Vector Machines.
\newblock \emph{ACM Transactions on Intelligent Systems and Technology}, 2:
  27:1--27:27.
\newblock Software available at \url{http://www.csie.ntu.edu.tw/~cjlin/libsvm}.

\bibitem[{Chatterjee, Chen, and Banerjee(2014)}]{chatterjee2014}
Chatterjee, S.; Chen, S.; and Banerjee, A. 2014.
\newblock Generalized dantzig selector: Application to the k-support norm.
\newblock \emph{Advances in Neural Information Processing Systems}, 27.

\bibitem[{Chevalier et~al.(2021)Chevalier, Nguyen, Salmon, Varoquaux, and
  Thirion}]{chevalier2021decoding}
Chevalier, J.-A.; Nguyen, T.-B.; Salmon, J.; Varoquaux, G.; and Thirion, B.
  2021.
\newblock Decoding with confidence: Statistical control on decoder maps.
\newblock \emph{NeuroImage}, 234: 117921.

\bibitem[{Fang et~al.(2020)Fang, Fan, Sun, and Friedlander}]{fang2020}
Fang, H.; Fan, Z.; Sun, Y.; and Friedlander, M. 2020.
\newblock Greed meets sparsity: Understanding and improving greedy coordinate
  descent for sparse optimization.
\newblock In \emph{International Conference on Artificial Intelligence and
  Statistics}, 434--444. PMLR.

\bibitem[{Farrens et~al.(2020)Farrens, Grigis, El~Gueddari, Ramzi, Chaithya,
  Starck, Sarthou, Cherkaoui, Ciuciu, and Starck}]{farrens2020pysap}
Farrens, S.; Grigis, A.; El~Gueddari, L.; Ramzi, Z.; Chaithya, G.; Starck, S.;
  Sarthou, B.; Cherkaoui, H.; Ciuciu, P.; and Starck, J.-L. 2020.
\newblock PySAP: Python Sparse Data Analysis Package for multidisciplinary
  image processing.
\newblock \emph{Astronomy and Computing}, 32: 100402.

\bibitem[{Foucart and Rauhut(2013)}]{foucart2013}
Foucart, S.; and Rauhut, H. 2013.
\newblock An invitation to compressive sensing.
\newblock In \emph{A mathematical introduction to compressive sensing}, 1--39.
  Springer.

\bibitem[{Gkirtzou et~al.(2013)Gkirtzou, Honorio, Samaras, Goldstein, and
  Blaschko}]{gkirtzou2013fmri}
Gkirtzou, K.; Honorio, J.; Samaras, D.; Goldstein, R.; and Blaschko, M.~B.
  2013.
\newblock fMRI analysis of cocaine addiction using k-support sparsity.
\newblock In \emph{2013 IEEE 10th International Symposium on Biomedical
  Imaging}, 1078--1081. IEEE.

\bibitem[{Golub and Van~Loan(2013)}]{golub2013matrix}
Golub, G.~H.; and Van~Loan, C.~F. 2013.
\newblock \emph{Matrix computations}.
\newblock JHU press.

\bibitem[{Golub et~al.(1999)Golub, Slonim, Tamayo, Huard, Gaasenbeek, Mesirov,
  Coller, Loh, Downing, Caligiuri et~al.}]{golub1999molecular}
Golub, T.~R.; Slonim, D.~K.; Tamayo, P.; Huard, C.; Gaasenbeek, M.; Mesirov,
  J.~P.; Coller, H.; Loh, M.~L.; Downing, J.~R.; Caligiuri, M.~A.; et~al. 1999.
\newblock Molecular classification of cancer: class discovery and class
  prediction by gene expression monitoring.
\newblock \emph{science}, 286(5439): 531--537.

\bibitem[{Grasmair, Scherzer, and Haltmeier(2011)}]{grasmair2011}
Grasmair, M.; Scherzer, O.; and Haltmeier, M. 2011.
\newblock Necessary and sufficient conditions for linear convergence of
  $\ell_1$-regularization.
\newblock \emph{Communications on Pure and Applied Mathematics}, 64(2):
  161--182.

\bibitem[{Harris et~al.(2020)Harris, Millman, Van Der~Walt, Gommers, Virtanen,
  Cournapeau, Wieser, Taylor, Berg, Smith et~al.}]{harris2020array}
Harris, C.~R.; Millman, K.~J.; Van Der~Walt, S.~J.; Gommers, R.; Virtanen, P.;
  Cournapeau, D.; Wieser, E.; Taylor, J.; Berg, S.; Smith, N.~J.; et~al. 2020.
\newblock Array programming with NumPy.
\newblock \emph{Nature}, 585(7825): 357--362.

\bibitem[{Haxby et~al.(2001)Haxby, Gobbini, Furey, Ishai, Schouten, and
  Pietrini}]{haxby2001distributed}
Haxby, J.~V.; Gobbini, M.~I.; Furey, M.~L.; Ishai, A.; Schouten, J.~L.; and
  Pietrini, P. 2001.
\newblock Distributed and overlapping representations of faces and objects in
  ventral temporal cortex.
\newblock \emph{Science}, 293(5539): 2425--2430.

\bibitem[{Jacob, Obozinski, and Vert(2009)}]{jacob2009group}
Jacob, L.; Obozinski, G.; and Vert, J.-P. 2009.
\newblock Group lasso with overlap and graph lasso.
\newblock In \emph{Proceedings of the 26th annual international conference on
  machine learning}, 433--440.

\bibitem[{Jain, Tewari, and Kar(2014)}]{Jain14}
Jain, P.; Tewari, A.; and Kar, P. 2014.
\newblock On Iterative Hard Thresholding Methods for High-dimensional
  M-Estimation.
\newblock In \emph{Advances in Neural Information Processing Systems},
  volume~27.

\bibitem[{Jia and Yu(2010)}]{jia2010}
Jia, J.; and Yu, B. 2010.
\newblock On model selection consistency of the elastic net when p >> n.
\newblock \emph{Statistica Sinica}, 595--611.

\bibitem[{Kakade et~al.(2009)Kakade, Shalev-Shwartz, Tewari
  et~al.}]{kakade2009duality}
Kakade, S.; Shalev-Shwartz, S.; Tewari, A.; et~al. 2009.
\newblock On the duality of strong convexity and strong smoothness: Learning
  applications and matrix regularization.
\newblock \emph{Unpublished Manuscript, http://ttic. uchicago.
  edu/shai/papers/KakadeShalevTewari09. pdf}, 2(1): 35.

\bibitem[{Kooperberg(1997)}]{kooperberg1997statlib}
Kooperberg, C. 1997.
\newblock StatLib: an archive for statistical software, datasets, and
  information.
\newblock \emph{The American Statistician}, 51(1): 98.

\bibitem[{Lang(1995)}]{lang1995newsweeder}
Lang, K. 1995.
\newblock Newsweeder: Learning to filter netnews.
\newblock In \emph{Machine Learning Proceedings 1995}, 331--339. Elsevier.

\bibitem[{Matet et~al.(2017)Matet, Rosasco, Villa, and Vu}]{matet2017}
Matet, S.; Rosasco, L.; Villa, S.; and Vu, B.~L. 2017.
\newblock Don't relax: early stopping for convex regularization.
\newblock \emph{arXiv preprint arXiv:1707.05422}.

\bibitem[{McDonald, Pontil, and
  Stamos(2016{\natexlab{a}})}]{mcdonald2016fitting}
McDonald, A.; Pontil, M.; and Stamos, D. 2016{\natexlab{a}}.
\newblock Fitting spectral decay with the k-support norm.
\newblock In \emph{Artificial Intelligence and Statistics}, 1061--1069. PMLR.

\bibitem[{McDonald, Pontil, and Stamos(2014)}]{mcdonald2014spectral}
McDonald, A.~M.; Pontil, M.; and Stamos, D. 2014.
\newblock Spectral k-support norm regularization.
\newblock \emph{Advances in neural information processing systems}, 27.

\bibitem[{McDonald, Pontil, and Stamos(2016{\natexlab{b}})}]{mcdonald2016new}
McDonald, A.~M.; Pontil, M.; and Stamos, D. 2016{\natexlab{b}}.
\newblock New perspectives on k-support and cluster norms.
\newblock \emph{The Journal of Machine Learning Research}, 17(1): 5376--5413.

\bibitem[{Molinari et~al.(2021)Molinari, Massias, Rosasco, and
  Villa}]{molinari2021iterative}
Molinari, C.; Massias, M.; Rosasco, L.; and Villa, S. 2021.
\newblock Iterative regularization for convex regularizers.
\newblock In \emph{International conference on artificial intelligence and
  statistics}, 1684--1692. PMLR.

\bibitem[{Moreau et~al.(2022)Moreau, Massias, Gramfort, Ablin, Bannier,
  Charlier, Dagr{\'e}ou, La~Tour, Durif, Dantas et~al.}]{moreau2022benchopt}
Moreau, T.; Massias, M.; Gramfort, A.; Ablin, P.; Bannier, P.-A.; Charlier, B.;
  Dagr{\'e}ou, M.; La~Tour, T.~D.; Durif, G.; Dantas, C.~F.; et~al. 2022.
\newblock Benchopt: Reproducible, efficient and collaborative optimization
  benchmarks.
\newblock In \emph{NeurIPS-36th Conference on Neural Information Processing
  Systems}.

\bibitem[{Natarajan(1995)}]{natarajan1995sparse}
Natarajan, B.~K. 1995.
\newblock Sparse approximate solutions to linear systems.
\newblock \emph{SIAM journal on computing}, 24(2): 227--234.

\bibitem[{Nesterov(2009)}]{nesterov2009}
Nesterov, Y. 2009.
\newblock Primal-dual subgradient methods for convex problems.
\newblock \emph{Mathematical programming}, 120(1): 221--259.

\bibitem[{Osher et~al.(2016)Osher, Ruan, Xiong, Yao, and Yin}]{osher2016}
Osher, S.; Ruan, F.; Xiong, J.; Yao, Y.; and Yin, W. 2016.
\newblock Sparse recovery via differential inclusions.
\newblock \emph{Applied and Computational Harmonic Analysis}, 41(2): 436--469.

\bibitem[{Pace and Barry(1997)}]{pace1997sparse}
Pace, R.~K.; and Barry, R. 1997.
\newblock Sparse spatial autoregressions.
\newblock \emph{Statistics \& Probability Letters}, 33(3): 291--297.

\bibitem[{Parikh, Boyd et~al.(2014)}]{parikh2014proximal}
Parikh, N.; Boyd, S.; et~al. 2014.
\newblock Proximal algorithms.
\newblock \emph{Foundations and trends{\textregistered} in Optimization}, 1(3):
  127--239.

\bibitem[{Pedregosa et~al.(2011)Pedregosa, Varoquaux, Gramfort, Michel,
  Thirion, Grisel, Blondel, Prettenhofer, Weiss, Dubourg
  et~al.}]{pedregosa2011scikit}
Pedregosa, F.; Varoquaux, G.; Gramfort, A.; Michel, V.; Thirion, B.; Grisel,
  O.; Blondel, M.; Prettenhofer, P.; Weiss, R.; Dubourg, V.; et~al. 2011.
\newblock Scikit-learn: Machine learning in Python.
\newblock \emph{the Journal of machine Learning research}, 12: 2825--2830.

\bibitem[{Penrose(1956)}]{penrose1956best}
Penrose, R. 1956.
\newblock On best approximate solutions of linear matrix equations.
\newblock In \emph{Mathematical Proceedings of the Cambridge Philosophical
  Society}, volume~52, 17--19. Cambridge University Press.

\bibitem[{Rhee et~al.(2006)Rhee, Taylor, Wadhera, Ben-Hur, Brutlag, and
  Shafer}]{rhee2006genotypic}
Rhee, S.-Y.; Taylor, J.; Wadhera, G.; Ben-Hur, A.; Brutlag, D.~L.; and Shafer,
  R.~W. 2006.
\newblock Genotypic predictors of human immunodeficiency virus type 1 drug
  resistance.
\newblock \emph{Proceedings of the National Academy of Sciences}, 103(46):
  17355--17360.

\bibitem[{Rockafellar(1970)}]{Rockafellar70}
Rockafellar, R.~T. 1970.
\newblock Convex Analysis.

\bibitem[{Scheetz et~al.(2006)Scheetz, Kim, Swiderski, Philp, Braun, Knudtson,
  Dorrance, DiBona, Huang, Casavant et~al.}]{scheetz2006regulation}
Scheetz, T.~E.; Kim, K.-Y.~A.; Swiderski, R.~E.; Philp, A.~R.; Braun, T.~A.;
  Knudtson, K.~L.; Dorrance, A.~M.; DiBona, G.~F.; Huang, J.; Casavant, T.~L.;
  et~al. 2006.
\newblock Regulation of gene expression in the mammalian eye and its relevance
  to eye disease.
\newblock \emph{Proceedings of the National Academy of Sciences}, 103(39):
  14429--14434.

\bibitem[{Tibshirani(1996)}]{tibshirani1996}
Tibshirani, R. 1996.
\newblock Regression shrinkage and selection via the lasso.
\newblock \emph{Journal of the Royal Statistical Society: Series B
  (Methodological)}, 58(1): 267--288.

\bibitem[{Tropp and Gilbert(2007)}]{tropp2007}
Tropp, J.~A.; and Gilbert, A.~C. 2007.
\newblock Signal recovery from random measurements via orthogonal matching
  pursuit.
\newblock \emph{IEEE Transactions on information theory}, 53(12): 4655--4666.

\bibitem[{Vanschoren et~al.(2014)Vanschoren, Van~Rijn, Bischl, and
  Torgo}]{vanschoren2014openml}
Vanschoren, J.; Van~Rijn, J.~N.; Bischl, B.; and Torgo, L. 2014.
\newblock OpenML: networked science in machine learning.
\newblock \emph{ACM SIGKDD Explorations Newsletter}, 15(2): 49--60.

\bibitem[{Vaskevicius, Kanade, and Rebeschini(2019)}]{vaskevicius2019}
Vaskevicius, T.; Kanade, V.; and Rebeschini, P. 2019.
\newblock Implicit regularization for optimal sparse recovery.
\newblock \emph{Advances in Neural Information Processing Systems}, 32.

\bibitem[{Wang et~al.(2015)Wang, Kwon, Li, and Shim}]{wang2015recovery}
Wang, J.; Kwon, S.; Li, P.; and Shim, B. 2015.
\newblock Recovery of sparse signals via generalized orthogonal matching
  pursuit: A new analysis.
\newblock \emph{IEEE Transactions on Signal Processing}, 64(4): 1076--1089.

\bibitem[{Wright and Ma(2022)}]{Wright2022}
Wright, J.; and Ma, Y. 2022.
\newblock \emph{High-Dimensional Data Analysis with Low-Dimensional Models:
  Principles, Computation, and Applications}.
\newblock Cambridge University Press.

\bibitem[{Xiao(2009)}]{Xiao09}
Xiao, L. 2009.
\newblock Dual Averaging Method for Regularized Stochastic Learning and Online
  Optimization.
\newblock In \emph{Advances in Neural Information Processing Systems},
  volume~22.

\bibitem[{Zhao, Yang, and He(2022)}]{zhao2022}
Zhao, P.; Yang, Y.; and He, Q.-C. 2022.
\newblock High-dimensional linear regression via implicit regularization.
\newblock \emph{Biometrika}.

\bibitem[{Zou and Hastie(2005)}]{zou2005regularization}
Zou, H.; and Hastie, T. 2005.
\newblock Regularization and variable selection via the elastic net.
\newblock \emph{Journal of the royal statistical society: series B (statistical
  methodology)}, 67(2): 301--320.

\end{thebibliography}

\appendix
\onecolumn

\section{Notations and definitions}
\label{sec:notations}

First, we describe some of the notations that will be used in this Appendix. $[\bm{v}]_{S}$ denotes the restriction of a vector $\bm{v}$ to the support $S$, $[\bm{v}]_{i}$ denotes its $i$-th component, $\bm{M}^{\top}$ denotes the transpose of a matrix $\bm{M}$, and $\bm{M}^{\dagger}$ denotes the Moore-Penrose pseudo-inverse of $\bm{M}$ \cite{golub2013matrix}. $\bm{I}_{r \times r}$ denotes the identity matrix in $\R^{r, r}$. $[d]$ denotes the set $\{1, ..., d\}$, and $[_{k}^{d}]$ denotes the set of all the sets of $k$ elements from $\{1, ..., d\}$.
 $\bar{S}$ denotes the complement in $[d]$ of a support $S$, that is, all the integers from $[d]$ that are not in $S$. $\partial f$ denotes the \textit{subgradient} of a function $f$\cite{Rockafellar70}. $\text{conv} (\mathcal{\mathcal{A}})$ denotes the convex hull of a set of vectors $\mathcal{A} \subset \R^d$ (that is, the set of all the convex combinations of elements of $\mathcal{A}$). We then introduce the following definitions:

  \begin{appdefinition}[Legendre-Fenchel dual \cite{Rockafellar70}]\label{def:dual}
    For any function $f: \R^d \rightarrow \R \cup \{-\infty , +\infty\}$, the function $f^*: \R^d \rightarrow \R$ defined by
    $$ f^*(\y):= \sup_{\w}\{\langle \y, \w \rangle - f(\w)\}$$ is the Fenchel \textit{conjugate} or \textit{dual} to $f$.
  \end{appdefinition}

   \begin{appdefinition}[hard-thresholding operator \cite{blumensath2009}]
      We define the \textit{hard-thresholding} operator for all $\z \in \R^d$ as the set $\pi_{HT}(\z) \subset \R^d$ below:
   $$   \pi_{HT}(\z) : = \arg\min_{\w \in \R^d ~ \text{s.t.} \|\w\|_0 \leq k} \|\w - \z\|_2^2 $$
 \end{appdefinition}

 \begin{appremark}
   $\pi_{HT}(\z)$ keeps the $k$-largest values of $\z$ in magnitude: but if there is a tie between some values, several solutions exist to the problem above, and the set  $\pi_{HT}(\z)$ is not a singleton.
 \end{appremark}

 \begin{appexample}
With $k=1$:  $\pi_{HT}((2, 1)) = \{(2, 0)\}$ and  $\pi_{HT}((2, 2)) = \{(2, 0), (0, 2)\}$   
 \end{appexample}

\begin{appdefinition}[top-$k$ norm]\label{def:topk}
  We define the following top-$k$ norm $\|\cdot \|_{(k)}$, for all $\w \in \R^d$:
  $$ \|\w \|_{(k)}  = \|\pi^*_{HT}(\w)\|_2$$
  Where $\pi^*_{HT}(\w)$ denotes any element from  $\pi_{HT}(\w)$ (since they all have the same norm). In other words, $\|\w \|_{(k)} $ is the $\ell_2$ norm of the top-$k$ elements from $\w$.
\end{appdefinition}
 
\section{Recall on the conditions of recovery with $\ell_1$ regularization}\label{app:recall}
In this section, we briefly recall a conditions for sparse recovery with $\ell_1$ norm regularization from \citet{grasmair2011}, and why it is implied by Assumption \ref{ass:assl1}.  The authors of \cite{grasmair2011} proved in their Theorem 4.7 that such Assumption \ref{ass:gras} below is a necessary and sufficient condition for achieving a linear rate of convergence for Tikhonov regularization with a priori parameter choice. We present below such condition \ref{ass:gras}.

\begin{appassumption}[Cond. 4.3 \cite{grasmair2011}]\label{ass:gras}~
  \begin{enumerate}
    \item $\w^*$ solves the equation $\XX \w = \bm{y}$
 \item  \textit{ Strong source condition:} There exist some $\bm{\lambda} \in \R^n$ such that :
   $$ \text{(i):}~ \XX^{\top} \bm{\lambda} \in \partial \| \cdot\|_1 (\w^*) \quad \text{and (ii):} \quad |\langle \bmx_i, \bm{\lambda} \rangle | < 1 \quad \text{for} \quad i \not\in \text{supp}(\w^*) $$
  where $\text{supp}(\w^*)$ is the support of $\w^*$ (that is, the set of the coordinates of its nonzero elements)
  \item \textit{ Restricted injectivity:}  The restricted mapping $\XX_{\text{supp}(\w^*)}$ is injective.
  \end{enumerate}
\end{appassumption}
We now show that this Assumption \ref{ass:gras} is implied by Assumption \ref{ass:assl1}:

\begin{applemma}
  Assumption \ref{ass:assl1} $\implies$  Assumption \ref{ass:gras}.
\end{applemma}

\begin{proof}
 \textbf{Assume Assumption \ref{ass:assl1}}, and 
  take $\bm{\lambda} = (\XX_{S}^{\dagger})^{\top}\text{sign}(\w^*_S)$.
 We now have the following equality (A): $\XX_S^{\top} \bm{\lambda} = \XX_S^{\top} (\XX_{S}^{\dagger})^{\top}\text{sign}(\w^*_S) = (\XX_{S}^{\dagger}\XX_S)^{\top}\text{sign}(\w^*_S)  \overset{(a)}{=}  \text{sign}(\w^*_S) = [\partial \|\cdot  \|_1(\w^*)]_S$ where (a) follows by property of the pseudo-inverse and the fact that $\XX_S$ is injective (and therefore full column rank).

 Additionally, from condition 3 in Assumption \ref{ass:assl1}, we have:
  $$\max_{\ell \in \bar{S}}|\langle\XX_S^{\dagger}\bmx_\ell, \text{sgn}(\w^*_S) \rangle |<1 \implies \max_{\ell \in \bar{S}}|\langle\bmx_\ell, (\XX_S^{\dagger})^{\top}\text{sgn}(\w^*_S) \rangle |<1 \implies \max_{\ell \in \bar{S}}|\langle\bmx_\ell,\bm{\lambda}\rangle |<1 $$
  This inequality above corresponds to (ii) from the \textit{strong source condition} above (\ref{ass:gras} (2. (ii))). Therefore, (since that last inequality also implies that for all  $i \not\in S,\langle \bmx_\ell,\bm{\lambda} \rangle \in [-1, 1] = [\partial \| \cdot \|_1(\w^*)]_i $, which, combined with (A) implies \ref{ass:gras} (2. (i)), we finally have that this $\bm{\lambda}$ verifies the existence conditions from \ref{ass:gras}.
 
\end{proof}

\section{Proof of Theorem \ref{thm:thm}}
\label{proof:thm}

\begin{proof}[Proof of Theorem \ref{thm:thm}]
  Theorem \ref{thm:thm} follows by combining Lemma \ref{lem:sub} with Theorem \ref{thm:matet} from \cite{matet2017}: in particular, when plugging from Lemma \ref{lem:sub} the value (denoted by $\bm{\lambda}^*$ in Lemma \ref{lem:sub}  (2)) of the solution of the dual  problem of~\eqref{eq:IKSNN}, (denoted by $\bm{v}^*$ in Theorem \ref{thm:matet}) we obtain:
  $$ \|\bm{v}^* \| = \| (\XX_S^{\top})^{\dagger} \w_S^*\|$$
\end{proof}

\begin{applemma}\label{lem:sub}
  Under Assumptions \ref{ass:sol} and \ref{ass:ass}, we have, with  $\bm{\lambda}^* := - (\XX_S ^{\top})^{\dagger}\w_S^*$:
  \begin{enumerate}[label={(\arabic*)}]
  \item  $- \XX^{\top}\bm{\lambda^*} \in \partial R(\w^*)$
  \item  $\bm{\lambda}^*$ is solution to the dual problem  of the noiseless problem below:
    \begin{align}
 \text{($I_{ks}$-noiseless)}: \quad \quad &\min_{\w} R(\w) \nonumber\\
\text{s.t.} \quad & \XX \w = \y  \label{eq:IKSNN}
\end{align}
  \end{enumerate}
\end{applemma}
\begin{proof}

  \textbf{Proof of (1):}\\

  We start by re-writing the condition $- \XX^{\top}\bm{\lambda} \in \partial R(\w^*)$ (for any given $\bm{\lambda}$) into a form easier to check:

  First, recall that $R(\w) =  \frac{1 - \alpha}{2}{\|\w \|^{sp}_{k}}^2  + \frac{\alpha}{2} {\|\w \|_2}^2 $.  We then have, for any $\bm{\lambda} \in \R^{n}$:
  \begin{align}
 \{-  \XX^{\top}\bm{\lambda} \in \partial R(\w^*) \}
                                                                 &\overset{}{\iff}
                                                                    \{ (1 - \alpha) \partial ( \frac{1}{2}{\|\cdot\|_{k}^{sp}}^2 )(\w^*)  \ni  - \XX^{\top} \bm{\lambda} - \alpha \w^*   \} \label{eq:thisone} \\
                                                                     &\overset{(a)}{\iff}
                                                                     \{ (1 - \alpha) \w^* \in \partial (\frac{1}{2} \|  \cdot \|_{(k)}^2 )(- \XX^{\top} \bm{\lambda} -  \alpha \w^*)  \}  \label{eq:bothsides} \\
                                                                         & \overset{(b)}{\iff} 
                                                                   \{   (1  -  \alpha ) \w^* \in \text{conv}(\pi_{HT}(-  \XX^{\top} \bm{\lambda} -  \alpha \w^*))  \} \label{eq:condtover}
  \end{align}
Where  $(a)$ follows from Proposition \ref{prop:reci} and Corollary \ref{cor:dualsq}, and $(b)$ from Lemma \ref{lem:subtopk}.

Let us now define $\bm{\lambda}^* :=-  (\XX_S ^{\top})^{\dagger}\w_S^*$. 
We then have:
\begin{align}
  \text{conv}(\pi_{HT}(- \XX^{\top} \bm{\lambda}^* - \alpha \w^*)) &= \text{conv}(\pi_{HT}(\XX^{\top} (\XX_S ^{\top})^{\dagger}\w_S^*   - \alpha \w^*)) \label{eq:befsimp}
         \end{align}
We now use the fact that :
\begin{enumerate}[label={(\Alph*)}]
\item  $\max_{\ell \in \bar{S}}| \langle \XX_S^{\dagger} \bmx_{\ell}, \w^*_{S}\rangle| < \min_{j \in S}| \langle \XX_S^{\dagger} \bmx_{j}, \w^*_{S}\rangle| $
  \item   $0< \alpha <  \frac{ \min_{j \in S}| \langle \XX_S^{\dagger} \bmx_{j}, \w^*_{S}\rangle| -  \max_{\ell \in \bar{S}}| \langle \XX_S^{\dagger} \bmx_{\ell}, \w^*_{S}\rangle| }{\|\w^*\|_{\infty}}$ (from the choice of $\alpha$ described in Theorem \ref{thm:thm})
\end{enumerate}
Which implies, for all $i \in S$, that:
\begin{align*}
  [|\XX (\XX_S^{\top})^{\dagger} \w^*_S  - \alpha \w^*|]_i &\overset{(a)}{\geq}  [|\XX (\XX_S^{\top})^{\dagger} \w^*_S|  - |\alpha \w^*|]_i \\
                                                              &=   [|\XX (\XX_S^{\top})^{\dagger} \w^*_S|]_i  - \alpha [|\w^*|]_i \\
                                                              & \geq    [|\XX (\XX_S^{\top})^{\dagger} \w^*_S|]_i  - \alpha \|\w^*\|_{\infty} \\
                                                         & \overset{(b)}{>}  [|\XX (\XX_S^{\top})^{\dagger} \w^*_S|]_i   -  \min_{j \in S}| \langle \XX_S^{\dagger} \bmx_{j}, \w^*_{S}\rangle| + \max_{\ell \in \bar{S}}| \langle \XX_S^{\dagger} \bmx_{\ell}, \w^*_{S}\rangle|\\
                                                         & \geq \max_{\ell \in \bar{S}}| \langle \XX_S^{\dagger} \bmx_{\ell}, \w^*_{S}\rangle|\\
                                                         &= \max_{\ell \in \bar{S}}  [|\XX (\XX_S^{\top})^{\dagger} \w^*_S  |]_{\ell}\\
    &\overset{(c)}{=} \max_{\ell \in \bar{S}}  [|\XX (\XX_S^{\top})^{\dagger} \w^*_S  - \alpha \w^*|]_{\ell}       
\end{align*}
Where  (a) follows from the reverse triangle inequality, (b) follows from (B), and (c) follows from the fact that the support of $\w^*$ is $S$ (so: $\forall j \in \bar{S}: ~ \w^*_j = 0$).

Therefore, for all $i \in S, \ell \in \bar{S}$:
$$ [|\XX (\XX_S^{\top})^{\dagger} \w^*_S  - \alpha \w^*|]_i >  [|\XX (\XX_S^{\top})^{\dagger} \w^*_S  - \alpha \w^*|]_{\ell} $$

This allows us to simplify \eqref{eq:befsimp}, given that the hard-thresholding operation selects the top $k$-components of a vector (in absolute value), and using the fact that we assumed that $S$ is of size $k$ (i.e. $|S|=k$) (so the $\text{conv}$ operation disappears here because since the inequality above is strict, there are no
``ties'' when computing the top-$k$ components (in absolute value); in other words, the convex hull of a singleton is that singleton itself): 

Therefore, for all $i \in [d]$:

\begin{align}
  [\text{conv}\left(\pi_{HT}(- \XX^{\top} \bm{\lambda}^* - \alpha \w^*)\right)]_i &=  \left\{   \begin{array}{l}
                                                            \langle \bmx_i, (\XX_S ^{\top})^{\dagger}\w_S^* \rangle  - \alpha \w_i^* \quad \text{if} \quad i \in S     \\
                                                               0 \quad \text{if}  \quad i \in \bar{S}\\
                                                                                             \end{array}\right.\\
  &= \left\{   \begin{array}{l}
                                                            \langle \bmx_i,  (\XX_S ^{\top})^{\dagger}\XX_S^{\dagger} \bm{y}  \rangle  - \alpha \w_i^* \quad \text{if} \quad i \in S     \\
                                                               0 \quad \text{if}  \quad i \in \bar{S}\\
               \end{array}\right.\\
  &\overset{(a)}{=} \left\{   \begin{array}{l}
                                                            [\XX_S^{\dagger} \bm{y}]_i  - \alpha \w_i^* \quad \text{if} \quad i \in S     \\
                                                               0 \quad \text{if}  \quad i \in \bar{S}\\
               \end{array}\right.\\
    &\overset{(b)}{=} \left\{   \begin{array}{l}
                                                            \w^*_i  - \alpha \w_i^* \quad \text{if} \quad i \in S     \\
                                                               0 \quad \text{if}  \quad i \in \bar{S}\\
                 \end{array}\right.\\
      &= \left\{   \begin{array}{l}
                                                            (1 - \alpha)\w^*_i  \quad \text{if} \quad i \in S     \\
                                                               0 \quad \text{if}  \quad i \in \bar{S}\label{eq:last}
                                                                   \end{array}\right.
  \end{align}

  Where (a) follows  from the following property of the pseudo-inverse for a matrix $\bm{M}$, applied to $\bm{M} = \XX_{S}^{\top}$: $\bm{M}\bm{M}^{\dagger} (\bm{M}^{\top})^{\dagger} = (\bm{M}^{\top})^{\dagger}$.
 (This property can be understood using the Singular Value Decomposition (SVD) expression for the pseudo-inverse \cite{golub2013matrix}: with $\bm{M} = \bm{U} \bm{D} \bm{V}^{\top}$, we have: $\bm{M}\bm{M}^{\dagger} (\bm{M}^{\top})^{\dagger} = \bm{U} \bm{D} \bm{V}^{\top}\bm{V} \bm{D}^{-1} \bm{U}^{\top} \bm{U} \bm{D}^{-1} \bm{V}^{\top} = \bm{U} \bm{D}^{-1}\bm{V}^{\top} = (\bm{M}^{\top})^{\dagger}$), and (b) follows from the fact that   $\w^*$ is the min $\ell_2$ norm solution on its support $S$ (as we assumed in Assumption \ref{ass:sol}), so $\XX_{S}^{\dagger} \bm{y}   = \w_S^*$ (III, 2, Corr. 3, \cite{ben2003generalized}, \cite{penrose1956best}).

Therefore, aggregating \eqref{eq:last} for all indices, we finally obtain: $$ \text{conv}\left(\pi_{HT}(- \XX^{\top} \bm{\lambda}^* - \alpha \w^*)\right) = (1 - \alpha) \w^*$$
That is, $\bm{\lambda}^*$ verifies \eqref{eq:condtover}.

So to sum up, under Assumptions \ref{ass:ass} and \ref{ass:sol}, we have that, for $\bm{\lambda}^* := -  (\XX_S ^{\top})^{\dagger}\w_S^*$: $- \XX^{\top}\bm{\lambda}^* \in \partial R(\w^*) $.

\textit{Note:} In addition, since  \eqref{eq:condtover} is equivalent to \eqref{eq:thisone}, plugging that value of $\bm{\lambda}^*$ into \eqref{eq:thisone}  we also have:

\begin{equation}
  \label{eq:useful}
 (1 - \alpha) \partial ( \frac{1}{2}{\|\cdot\|_{k}^{sp}}^2 ) (\w^*)  \ni  \XX^{\top} (\XX_S ^{\top})^{\dagger} \w_S^* - \alpha \w^*   
\end{equation}

(This latter equation will be useful in the proof of (2) below)

\textbf{Proof of (2):}\\

We now turn to proving the second part (i.e. (2)) of  Lemma \ref{lem:sub}.

  As described in \cite{matet2017},  the dual problem of ~\eqref{eq:IKSNN} can be written as (see e.g. Definition 15.19 in \cite{bauschke2011convex}):
  \begin{equation}\label{eq:dualpb}
\min_{\vv} R^*(-\XX^{\top} \vv) + \langle \bm{y}, \vv\rangle    
  \end{equation}

  where $R^*$ denotes the Fenchel Dual of $R$ (see Definition \ref{def:dual} ). 

  Let us define, for all $\bm{v} \in \R^{n}$: $f(\bm{v}) = R^*(-\XX^{\top} \bm{v})$

  The first order optimality condition of problem \eqref{eq:dualpb} can be written as:
  $$ \partial f(\bm{v}) + \y \ni \bm{0}$$
  Which is equivalent to:
    $$ - \partial f(\bm{v}) \ni  \y$$

  Therefore, if we find $\bm{v}$ such that the expression above is verified, then that $\bm{v}$ is solution of \eqref{eq:dualpb}.

  Now, from Theorem 23.9 in \cite{Rockafellar70}, we have that: $-\XX \partial R^*(-\XX^{\top} \vv) \subset \partial f(\bm{v}) $ (that is, the subgradient verifies a similar chain rule as the usual gradient, in one direction of inclusion).

  Note now that since $R$ is $\alpha$-strongly convex (due to the squared $\ell_2$ norm term), $R^*$ is differentiable and $\alpha$-smooth \cite{kakade2009duality} and therefore, its gradient is well defined, so we can rewrite $\partial R^*$ into $\nabla R^*$ (the subgradient is a singleton).
  
  Now, take $\bm{v}^{*}: = -  (\XX_S ^{\top})^{\dagger}\w_S^*$.

  Let us compute $\nabla R^*(-\XX^{\top} \bm{v}^{*} ) =  \nabla R^*(\XX^{\top} (\XX_S ^{\top})^{\dagger}\w_S^*) $.

  Let us denote $ \bm{z} : = \nabla R^*(\XX^{\top}  (\XX_S ^{\top})^{\dagger}\w_S^*) $. From \ref{prop:reci}, we have the following equivalences:
  \begin{align}
    & \XX^{\top}  (\XX_S ^{\top})^{\dagger}\w_S^* \in \partial R(\z) \nonumber\\
    & \iff \XX^{\top}  (\XX_S ^{\top})^{\dagger}\w_S^*  \in (1 - \alpha) \partial (\frac{1}{2}{\|\cdot \|^{sp}_k}^2)(\z) + \alpha \z \nonumber\\
                & \iff \XX^{\top} (\XX_S ^{\top})^{\dagger}\w_S^*  -  \alpha \z \in (1 - \alpha) \partial (\frac{1}{2}{\|\cdot \|^{sp}_k}^2)(\z) \label{eq:tofill}
  \end{align}
  Now, we know from \eqref{eq:useful} that taking $\z := \w^*$  satisfies expression \eqref{eq:tofill}.
  Therefore:
  $  \nabla R^*(-\XX^{\top} \bm{v}^{*}) = \w^*$

  Now, we can see that the proof is complete, since we know from Assumption \ref{ass:sol} that $\bm{y} = \XX \w^*$. So using the above, we have:
  \begin{equation*}
    \bm{y} = \XX \w^* = \XX \nabla R^*(-\XX^{\top} \bm{v}^{*}) \in \{ \XX \nabla R^*(-\XX^{\top} \bm{v}^{*})\} = \XX \partial R^*(-\XX^{\top} \bm{v}^{*}) \subset - \partial f(\bm{v}^*)
  \end{equation*}
  So to sum up, we have that: $ \bm{y} \in - \partial f(\bm{v}^*) $,
  which means that $\bm{v}^* = -  (\XX_S ^{\top})^{\dagger}\w_S^*$ is solution  of  the dual problem of ~\eqref{eq:IKSNN}.

\end{proof}

\begin{apptheorem}[\cite{matet2017}]\label{thm:matet}Let $\delta\in\left]0,1\right]$ and let  $(\hat{\w}_t)_{t\in\mathbb{N}}$  be  the sequence generated by ADGD (cf. \cite{matet2017}).  Assume that there exists $\bm{\lambda} \in\mathbb{R}^n$ such that $-\XX^T\bm{\lambda}\in\partial R(\w^*)$.
Set $a=4\|\XX\|^{-1}$ and $b=2\|\XX\| \|\bm{v}^*\|/\alpha$, where $\bm{v}^*$ is a solution of the dual  problem of~\eqref{eq:IKSNN}. Then, for every $t\geq 2$,
\begin{equation}
\| \hat{\w}_t -\w^*\|  \leq a t\delta + b{t}^{-1}.
\end{equation}
In particular, choosing $t_\delta=\lceil c \delta^{-1/2}\rceil $ for some $c>0$,
\begin{equation}
\| \hat{\w}_t -\w^*\|  \leq  \big[a(c+1)+bc^{-1}\big] \delta^{1/2}.
\end{equation}
\end{apptheorem}
\begin{proof}
  Proof in \cite{matet2017}
\end{proof}

\section{Useful Results}
  Here we present some lemmas and theorems that are used in the proofs above:

  \begin{apptheorem}[Corollary 4.3.2, \cite{Baptiste01}]\label{thm:baptiste}
    Let $f_1, ..., f_m$ be $m$ convex functions from $\R^d$ to $\R$ and define
    $$ f:= \max \{ f_1 ,..., f_m\} .$$
    Denoting by $I(\w) := \{ i: f_i(\w) = f(\w)\}$ the active index-set, we have:
    $$ \partial f(\w) = \text{conv} (\cup \partial f_i(\w): i \in I(\w))$$
  \end{apptheorem}

  \begin{proof}
    Proof in \cite{Baptiste01}.
  \end{proof}
  
  \begin{applemma}[Subgradient of the half-squared top-$k$ norm]\label{lem:subtopk}
    Let $n$ be the (half-squared) top $k$-norm: $n(\w) = \frac{1}{2}\|\w\|^2_{(k)}$. We have:
    $$ \partial n(\w) = \text{conv}(\pi_{HT}(\w))$$
  \end{applemma}

  \begin{proof}
Let us denote each possible supports of $k$ coordinates from $[_{k}^{d}]$ by $\mathcal{I}_{i}$ for $i = 1, ..., \binom{d}{k}$ 
    The top-$k$ norm can be written as follows:
    $$n(\w) = \max_{i} n_i(x) = \max \{n_{1}(\w), ..., n_{\binom{d}{k}}(\w)\} $$ where each $n_{i} = \frac{1}{2}\|\w_{\mathcal{I}_i}\|_2^2$, with $\w_{\mathcal{I}_i}$  the thresholding of $\w$ with all coordinates not in $\mathcal{I}_i$ set to 0.
    Let us denote, for a given $\w \in \R^d$, $\Pi(\w) \subset [_{k}^{d}]$ to be the set of  \textit{supports} such that for any $j \in \Pi(\w)$: $n_{j}(\w) = n(\w)$. In other words, $\Pi(\w)$ denotes the \textit{active index set} described in Theorem \ref{thm:baptiste}. Those supports are those which select the top-$k$ components of $\w$ in absolute value (several choices are possible). In other words:
$$\pi_{HT}(\w) = \{ \w_{\mathcal{I}_j}: j \in \Pi(\w)\}$$
    Now, we know that for all $i \in \binom{d}{k}$, $n_i$ is differentiable, since $n_i$ is simply the half squared $\ell_2$ norm of the thresholding of $\w$ on a fixed support $\mathcal{I}_i$. Since it is differentiable, its subgradient is thus a singleton composed of its gradient: $\partial  n_{i} (\w) = \{\nabla n_{i} (\w) \}= \{\w_{\mathcal{I}_i}\}$. 
    
    Therefore, from Theorem \ref{thm:baptiste}, we have:
    
$$ \partial f(\w) = \text{conv} (\nabla f_i(\w): i \in \Pi(\w)) =  \text{conv} ( \w_{\mathcal{I}_j}: j \in \Pi(x))=  \text{conv} (\pi_{HT}(\w))$$

  \end{proof}

\begin{appproposition}[Proposition 11.3, \cite{Rockafellar70}]\label{prop:reci}
  For any proper, lsc, convex function $f$, denote by $f^*$ its Fenchel dual defined above in \ref{def:dual}. One has $\partial f^* = (\partial f)^{-1}$ and $\partial f = (\partial f^*)^{-1}$.
\end{appproposition}

\begin{proof}Proof in \cite{Rockafellar70}.
\end{proof}

\begin{applemma}[Fenchel conjugate of a half squared norm \cite{boyd2004convex} Example 3.27, p. 94]\label{lem:fench}
  Consider the function $f(\w) = \frac{1}{2}\|\w\|^2$, where $\|\cdot\|$ is a norm, with dual norm $\| \cdot \|_{*}$. Its Fenchel conjugate is $f^*(\w) = \frac{1}{2}\|\w\|_{*}^2$.
\end{applemma}

\begin{proof}
  Proof in \cite{boyd2004convex}.
\end{proof}

\begin{applemma}[Dual of the $k$-support norm, \cite{argyriou2012sparse}, 2.1]\label{lem:topkdual}
  Denote by $(\| \cdot \|)_{*}$ the dual norm of a norm $\| \cdot \|$. The top-$k$ norm (see Definition \ref{def:topk}) is the dual norm of the $k$-support :
  $$ (\| \cdot \|_{k}^{sp})_{*} =  \|\cdot \|_{(k)}$$
\end{applemma}

\begin{appcorollary}\label{cor:dualsq}
$$  (\frac{1}{2} {\| \cdot \|_{k}^{sp}}^2)^{*} =  \frac{1}{2} \| \cdot \|_{(k)}^2$$
\end{appcorollary}

\begin{proof}
Corollary  \ref{cor:dualsq} follows from Lemmas \ref{lem:fench} and \ref{lem:topkdual}.
\end{proof}

\section{Proximal operator of the $k$-support norm}\label{sec:proxksn}

In this section, we describe the method that we use to compute the proximal operator of the half-squared $k$-support norm, as is described in Algorithm 1 from \cite{mcdonald2016new}. In our code (available at \url{https://github.com/wdevazelhes/IRKSN_AAAI2024}), we use an existing implementation from the \texttt{modopt} package \cite{farrens2020pysap}. Note that Algorithm 1 from  \cite{mcdonald2016new} was originally described in a more general formulation, from which the algorithm described below can be obtained by fixing $a=0$, $b=1$, and $c=k$ (we refer the reader to \cite{mcdonald2016new} for more details on what variables $a$, $b$, and $c$ refer to).

\begin{algorithm}[ht]
\caption{(\cite{mcdonald2016new}, Algorithm 1) Computation of $\bm{x}= \text{prox}_{\frac{\lambda}{2}\|\cdot \|_{(k)}^2}\left(\w \right)$, with:  $\forall \alpha \in \mathbb{R}_+: S(\alpha) := \sum_{i=1}^d \min(1, \max(0, \alpha |w_i| - \lambda))$. }
\label{alg:prox_01}
\begin{algorithmic} 
\REQUIRE parameter: $\lambda$.
\STATE \textbf{1.} Sort points $\left\{ \alpha^i \right\}_{i=1}^{2d} = \left\{ \frac{\lambda}{ \vert w_j \vert}, \frac{1+\lambda}{ \vert w_j \vert} \right\}_{j=1}^d$  such that $\alpha^i \leq \alpha^{i+1}$.
\STATE \textbf{2.} Identify points $\alpha^i$ and $\alpha^{i+1}$ such that $S(\alpha^i) \leq k$ and $S(\alpha^{i+1})\geq k$ by binary search.
\STATE \textbf{3.} Find $\alpha^*$ between $\alpha^i$ and $\alpha^{i+1}$ such that $S(\alpha^*)=k$ by linear interpolation.
\STATE \textbf{4.} Compute $\theta_i(\alpha^*):= \min(1, \max(0, \alpha^* |w_i| - \lambda))$ for $i=1\ldots, d$.
\STATE \textbf{5.} Return $x_i =\frac{\theta_i w_i}{\theta_i+\lambda}$ for $i=1\ldots, d$.
\end{algorithmic}
\end{algorithm}

\section{Details on the synthetic experiment from Section \textit{Experiments}}

In this section, we present the hyperparameters for the experiments in Section \textit{Experiments}, for each algorithm. First, we fix $k=10$ for all algorithms that require setting a parameter $k$. We run the algorithms for a maximum number of iterations of 20,000. Note that in this synthetic experiment we do not fit the intercept or center the data since the data has $0$ mean. For IHT, we search $\eta$ in $\{0.0001, 0.001, 0.01, 0.1, 1.\}$. For Lasso, we use the implementation \texttt{lasso\_path} from \texttt{scikit-learn} \cite{pedregosa2011scikit}, with its default parameters, which automatically choses the path of  $\lambda$ based on a data criterion. For ElasticNet, we use the implementation \texttt{enet\_path} from \texttt{scikit-learn} \cite{pedregosa2011scikit}, which similarly as above, automatically chooses the path of  $\lambda$ based on a data criterion. In addition, we choose the recommended values $\{.1, .5, .7, .9, .95, .99, 1\}$ of \texttt{ElasticNetCV} for the relative weight of the $\ell_1$ penalty. For the KSN algorithm (i.e. linear regression penalized with $k$-support norm), we choose the strenght of the $k$-support norm penalty $\lambda$ in $\{0.1, 1.\}$, and we set $L$ (which is the inverse of the learning rate) to $1e6$ similarly as in Section \ref{app:details}. For OMP, we use the implementation from \texttt{scikit-learn} \cite{pedregosa2011scikit}. For SRDI, we search for the parameters $\kappa$ and $\alpha$ from \cite{osher2016}, respectively in the intervals $\{0.0001, 0.001, 0.01, 0.1, 1.\}$ and $\{0.0001, 0.001, 0.01, 0.1, 1.\}$. For IROSR,
we search for the parameters $\eta$ and $\alpha$ respectively in $\{0.0001, 0.001, 0.01, 0.1, 1.\}$ and $\{0.0001, 0.001, 0.01, 0.1, 1.\}$. For IRCR, we set $\tau$ and $\sigma$ to $\frac{0.9}{\sqrt{2 \|\bm{X}\|^2}}$ (in order to verify the condition of equation (6) in \cite{molinari2021iterative}) similarly as in section \ref{app:details}. For IRKSN (ours), we search $\alpha$ (from Algorithm \ref{alg:irksn}) in $\{0.0001, 0.001, 0.01, 0.1, 1, 10\}$. Our results are produced on a server of CPUs with 16 cores the experiment takes a few hours to run.

\section{Path of IRKSN vs Lasso vs ElasticNet}
\label{app:path}

In this section, we plot in Figure \ref{fig:pathapp} the path of ElasticNet (with an $\ell_1$ ratio of $0.8$, i.e. its penalty is $\lambda (0.8 \|\cdot\|_{1} + 0.2 \| \cdot \|_2^2)$), in addition to the plot of the Lasso path and the  IRKSN path, from Section \textit{Illustrating Example}. As we can see, the ElasticNet, as the Lasso, cannot recover the true sparse vector.

      \begin{figure}[!bht]
  \centering
  \subfigure[Lasso path]{\includegraphics[scale=0.4]{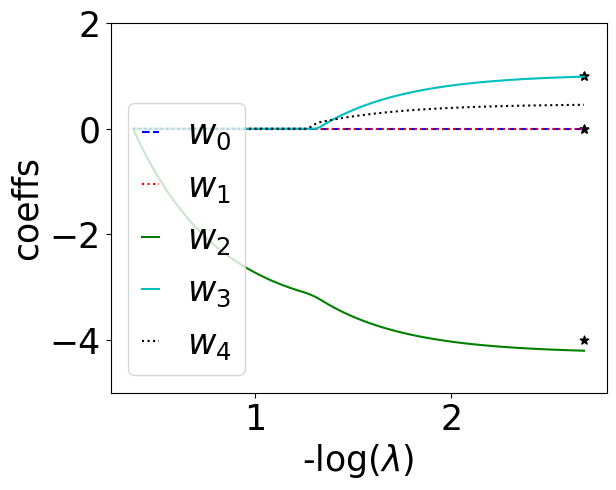}\label{fig:lassopathapp}} 
  \subfigure[IRKSN path]{\includegraphics[scale=0.4]{./figs/irksn_path.png}\label{fig:irksnpathapp}}
  \subfigure[ElasticNet path]{\includegraphics[scale=0.4]{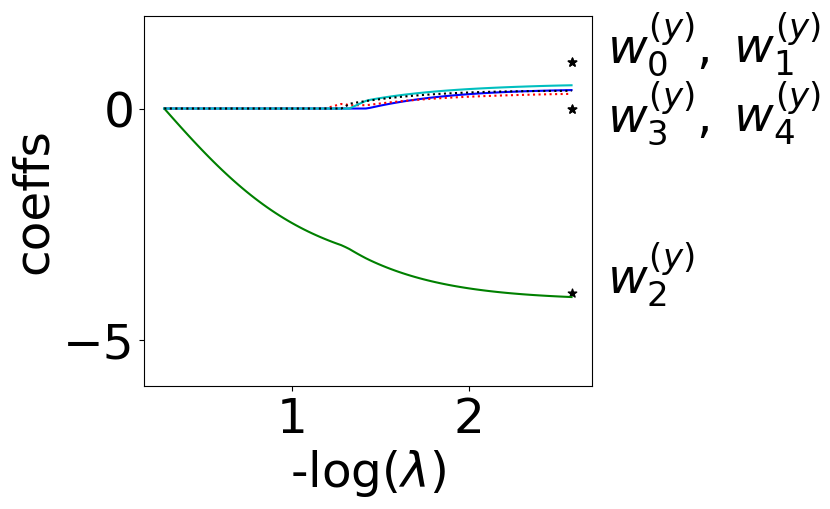}\label{fig:enetpathapp}}
   \caption{Comparison of the path of IRKSN with Lasso and Elasticnet}
   \label{fig:pathapp}
 \end{figure}

\section{Extra experiment: fMRI decoding}

\subsection{Setting}

\paragraph{Data-set construction} We consider a functional MRI (fMRI) decoding experiment, where observations $\X$ are activity recordings (3D activity voxel maps) of fMRI for several subjects which are presented with images of two different classes, i.e. where the observed target $\y^{\delta}$ comprises labels from the set $\{class1, class2\}$ converted to -1 and 1 respectively. It was shown experimentally in \cite{belilovsky15b} that $k$-support norm regularization (as a penalty) performs significantly better than Lasso on such kind of fMRI tasks: we therefore wish to evaluate whether this is true also for iterative regularization with $k$-support norm. We use the Haxby dataset \cite{haxby2001distributed}, downloaded with the use of the \texttt{nilearn} package \cite{Chamma_nilearn}. 
We then prepare the data from raw recordings following closely the protocol from the fMRI example from the package \texttt{hidimstat}\footnote{\url{https://ja-che.github.io/hidimstat/auto_examples/plot_fmri_data_example.html}}, choosing the neural recordings of a specific subject (subject number 2), as they do. Once we obtain the data matrix $\X$ and target $\y^{\delta}$, we use the algorithms from Table \ref{tab:pen} to estimate the true model $\w^*$, which is an estimate of the brain functional region associated with the true (noiseless) response variable $\y$. Such dataset contains 216 samples, of dimensionality 39912. We split the dataset into a training set and a validation set, with the ratio 80\%-20\%: since we consider only the support reconstruction task, we indeed do not use any test-set in this case.

\paragraph{Hyperparameters and algorithms tuning}
Below we give more details on the tuning of each algorithm. Once the dataset is prepared, we fine tune the algorithms hyperparameters on mean squared error prediction on the validation set. The intercept of the models is fitted separately, using the same method as in section \ref{app:details}. Additionally, we preprocess first the data by removing features of variance $0$ and centering and standardizing $\X$, as described in section \ref{app:details}.  Additionally, since such dataset of neural images is high dimensional, to reduce the computational cost we use sensible values for hyperparameters whenever those are possible: for instance, for algorithms that have convergence guarantees if the learning rate is equal to the inverse of the Lipschitz-smoothness constant (which in our case is the squared nuclear norm of $\X$ (denoted $\| \X \|^2 $)), we set the learning rate denoted by $\eta$ to such value. Also, for all algorithms which require setting a sparsity level $k$ (IRKSN, IHT, KSN, OMP), we set $k=150$, which is an estimate that we considered known \textit{a priori} for the size of the function region we wish to reconstruct. Additionally, we run all algorithms with a maximum number of iterations of 10,000. For IHT, we set $\eta=\frac{1}{\|\X \|^2}$. For Lasso, we use the implementation \texttt{lasso\_path} from \texttt{scikit-learn} \cite{pedregosa2011scikit}, with its default parameters, which automatically choses the path of  $\lambda$ based on a data criterion. For ElasticNet, we use the implementation \texttt{enet\_path} from \texttt{scikit-learn} \cite{pedregosa2011scikit}, which similarly as above, automatically chooses the path of  $\lambda$ based on a data criterion. In addition, we choose the recommended values $\{.1, .5, .7, .9, .95, .99, 1\}$ of \texttt{ElasticNetCV} for the relative weight of the $\ell_1$ penalty. For KSN penalty, we choose the strenght of the $k$-support norm penalty $\lambda$ in $\{0.1, 1.\}$, and set $\eta = \frac{1}{\| \X\|^2}$. For SRDI, we search for the parameters $\kappa$ and $\alpha$ from \cite{osher2016}, respectively in $\{0.001, 0.01, 0.1\}$ and $\{0.001, 0.01, 0.1\}$. For IROSR, we search for the parameters $\eta$ and $\alpha$ respectively in $\{0.001, 0.01, 0.1\}$ and $\{0.001, 0.01, 0.1\}$. For IRCR, we set $\tau$ and $\sigma$ to $\frac{0.9}{\sqrt{2 \|\bm{X}\|^2}}$ (in order to verify the condition of equation (6) in \cite{molinari2021iterative}). For IRKSN (ours), we set $\alpha$ (from Algorithm \ref{alg:irksn}) to $0.001$, since as recommended by Theorem \ref{thm:thm}, a smaller value of $\alpha$ has more chance to verify the conditions for convergence of \ref{thm:thm}, (assuming $\X$ verifies Assumptions \ref{ass:sol} and \ref{ass:ass}). Additionally, a smaller $\alpha$ ensures that the sparsity of the iterates is closer to $k$ sparse. Our results are produced on a server of CPUs with 16 cores the experiment takes a few hours to run.

\paragraph{Post processing:} Once the estimated model $\w$ is returned by each method, we post process $\w$ as in \cite{chevalier2021decoding} Section 3.2: we first compute the corresponding corrected $p$-values obtained when assumed that the weights values are sampled from a Gaussian distribution (see \cite{chevalier2021decoding} for more details). Then, we transform those as $z$-value maps instead of $p$ values maps, and set the FWER (Family-Wise Error Rate) threshold for detection to 0.1  as is done in \cite{chevalier2021decoding}, which is translated into a corresponding threshold for $z$-values using the Bonferroni correction. We refer the reader to \cite{chevalier2021decoding} for more details on such post-processing and the related terminology. We then plot in Figure \ref{fig:fmri} the estimated functional region for all of the methods.

\subsection{Results}

\paragraph{Visual comparison on reconstruction}
We plot the fMRI reconstruction results of each method on Figure \ref{fig:fmri}, in the case where $class1$ correspond to the class 'face' and $class2$ corresponds to the class 'house'. As a comparison, we also have plotted in Figure \ref{fig:encludl} the result of the EnCluDL algorithm from \cite{chevalier2021decoding} in the same setting, and which may be considered as a ground truth: such method indeed uses knowledge of the spatial structure of the voxel grid (i.e., which voxel is close to each voxel, therefore more likely to be correlated with it), contrary to the methods considered in our paper which are blind to such structure. As we can see, methods based on an implicit or explicit $\ell_1$ norm regularization perform poorly, since they tend to estimate a support that is too small: indeed, methods such as the Lasso are known to fail to select group of correlated column, and tend to select only a few explicative features \cite{zou2005regularization}. On the contrary, $k$-support norm regularization like IRKSN is able to estimate a support of larger size, which by inspection seems to be a better estimate of the ground truth. Additionally, we can observe that even ElasticNet, which supposedly should also be able to perform reasonably well in presence of correlated features \cite{zou2005regularization}, does not seem to recover the true functional region: indeed, although its $\ell_1$ and $\ell_2$ penalties are tuned by grid-search on a validation set, it is more difficult for such method to fix a specific sparsity $k$. On the other hand, methods which fix a specific $k$ (IHT, KSN, IRKSN, OMP) are advantaged. We can also notice that the solutions of IHT and IRKSN are almost the same, and appear to be the most successful reconstruction of the active functional brain region.
\paragraph{Quantitative results} 
Finally, we  also provide extra quantitative results in Table \ref{tab:quant} below for the face/house and house/shoe data splits, in terms of $\|\bm{w} - \bm{w}^* \|$, where for the ground truth $\w^*$ we take the weight vector obtained by running the EnCluDL method. Note that IHT also has a good performance, but, unlike IRKSN, the theory of sparse recovery for IHT fails to explain such success since (see, e.g. \cite{Jain14}) the RIP property is typically not verified for correlated data (like fMRI \cite{belilovsky15b}):
\begin{table}[H]
\centering
\begin{tabular}{cccccccccc}
\toprule
 & Lasso & ElasticNet & OMP & IHT & KSN & IRKSN & IRCR & IROSR & SRDI \\
\midrule
face'/'house' & .425 & .349 & .938 & .2441 & .247 & \textbf{.2440} & .341 & .381 & .314 \\
'house'/'shoe' & .528 & .500 & .938 & .2968 & .299 & \textbf{.2965} & .407 & .502 & .357 \\
\bottomrule
\end{tabular}
\caption{Comparison of the algorithms on model estimation $\| \bm{w} - \bm{w}^*\|$ ($\w^*$: weight vector obtained by running the EnCluDL method).}
\label{tab:quant}
\end{table}

      \begin{figure}[!htp]
  \centering
  \subfigure[Lasso]{\includegraphics[scale=0.4]{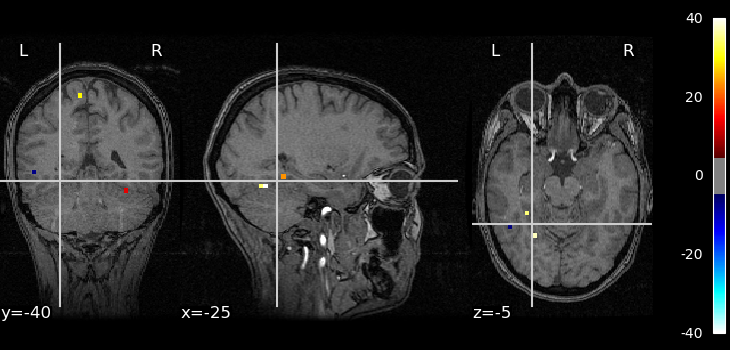}\label{fig:fmrilasso}} 
  \subfigure[ElasticNet]{\includegraphics[scale=0.4]{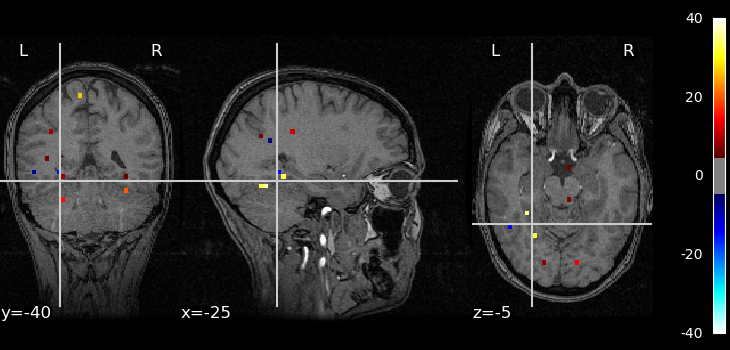}\label{fig:fmrienet}}
    \subfigure[OMP]{\includegraphics[scale=0.4]{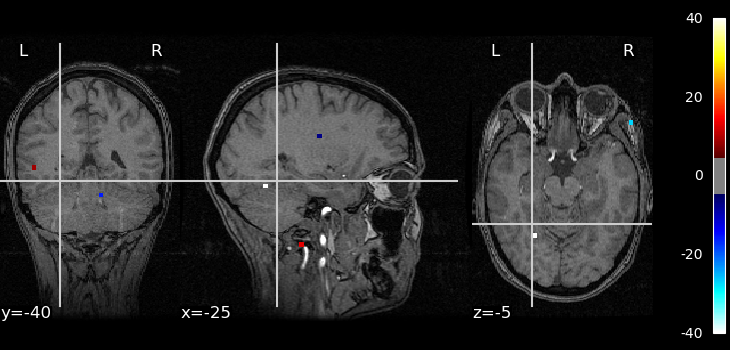}\label{fig:fmriomp}}
    \subfigure[SRDI]{\includegraphics[scale=0.4]{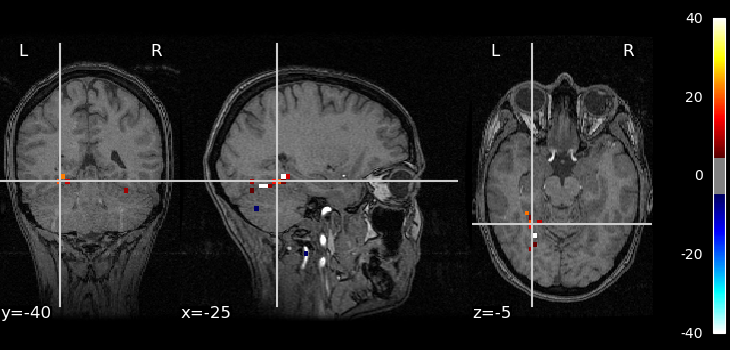}\label{fig:fmrisrdi}}
        \subfigure[IROSR]{\includegraphics[scale=0.4]{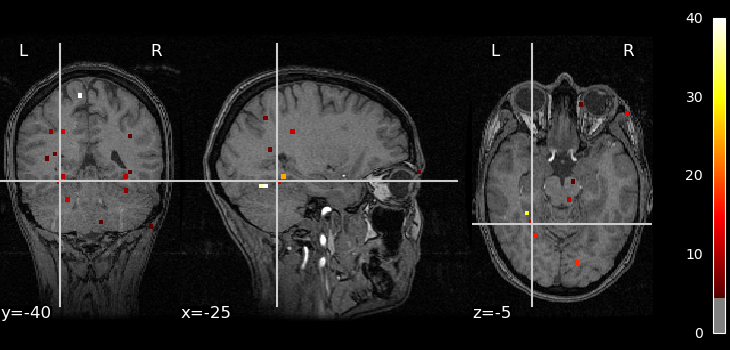}\label{fig:fmriirosr}}
    \subfigure[IHT]{\includegraphics[scale=0.4]{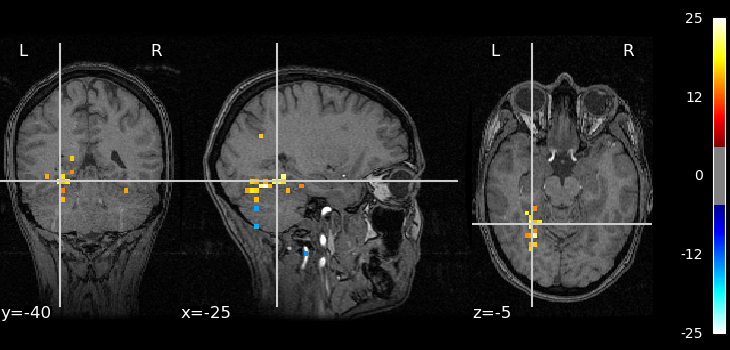}\label{fig:fmriiht}}
        \subfigure[IRCR]{\includegraphics[scale=0.4]{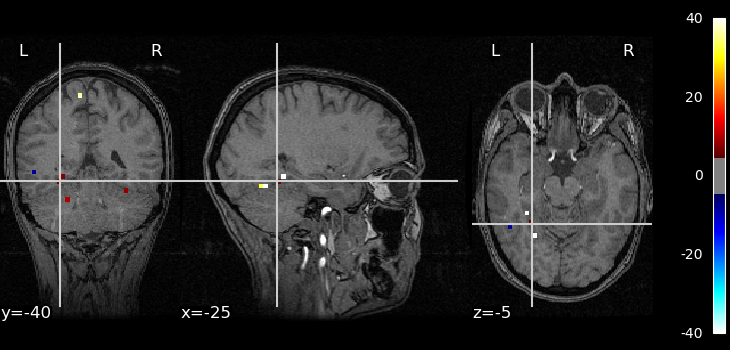}\label{fig:fmriircr}}
            \subfigure[KSN]{\includegraphics[scale=0.4]{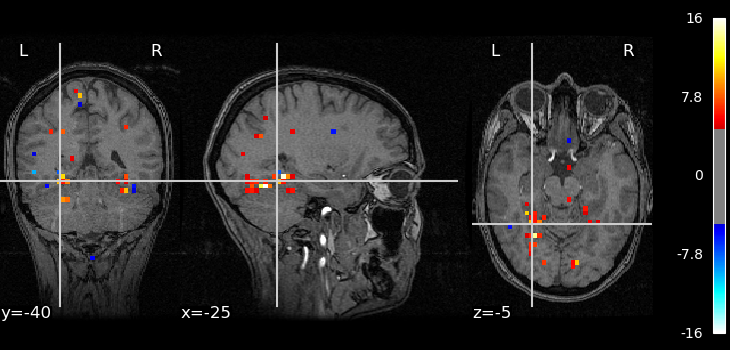}\label{fig:fmriksn}}
              \subfigure[IRKSN]{\includegraphics[scale=0.4]{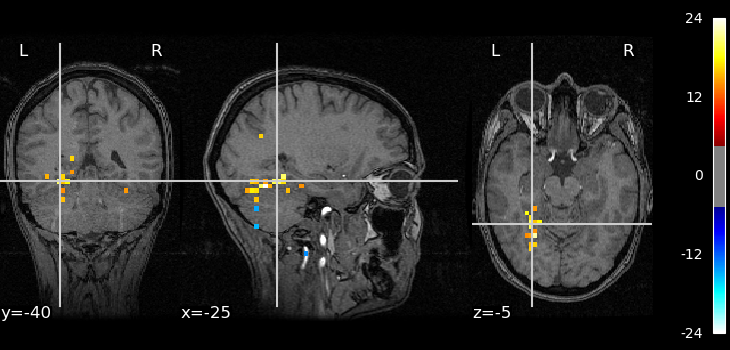}\label{fig:fmriirksn}}
                            \subfigure[EnCluDL]{\includegraphics[scale=0.4]{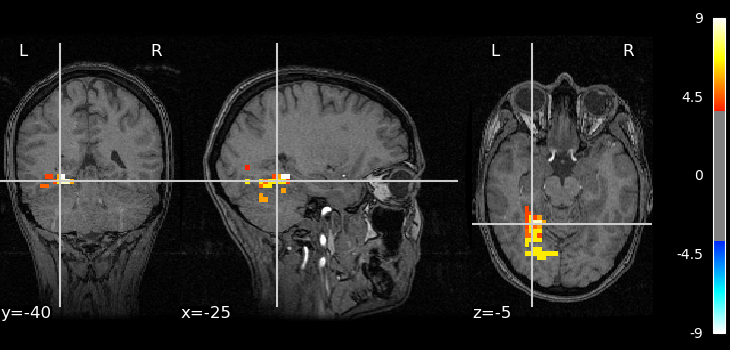}\label{fig:encludl}}
   \caption{Comparison of different methods on an fMRI decoding task. Figure \ref{fig:encludl} is the EnCluDL method from \cite{chevalier2021decoding} which uses the additional knowledge of the spatial structure of the voxels, and may be considered as some ground truth for the functional region to be reconstructed.}
   \label{fig:fmri}
 \end{figure}

\subsection{Interpretation}

Such an fMRI dataset is a real-life (non-synthetic) dataset, therefore its true data generating process is unknown. However, we provide here some attempt to explain the success of $k$-support norm regularization on such fMRI reconstruction task, in the light of our newly derived sufficient conditions for recovery derived in the paper. More precisely, we present below a data generating process that we believe might potentially be similar to the true underlying data generating process of fMRI observations, and which we will show actually verifies our Assumptions \ref{ass:sol} and \ref{ass:ass} for recovery with IRKSN.

\paragraph{Example 2: Simplistic fMRI generating process} For each observation $i$, $\x_i$ represents the observed activated voxels from the fMRI, so we may consider that they consist in (i) the functional region related to the noiseless target $y$ (i.e. in this case for instance, say, the visual cortex functional region), and (ii) some unrelated regions that are activated for some other reasons (e.g. the functional region responsible for movement if the subject is moving).

Therefore, we model each observation $\x_i$ (i.e. row of $\X$, seen as a column vector) as follows: 

$$\x_i = y_i \w^* + \bm{\gamma}_i$$

Where $\w^*$ is the true model, which support $S=\text{supp}(\w^*)$ is the true functional region we wish to reconstruct, $y_i$ is considered to be both the noiseless target variable, but also the variable modulating the functional region: for instance, if $y_i$ denotes the presence or absence of an image in front of the subject, the functional region for visual stimuli will be more or less active depending on $y_i$, and where $\bm{\gamma}_i$ is a variable which we consider to have a support disjoint from $supp(\w^*)$, which denotes all the other unrelated functional region that are active at observation $i$ (e.g. as discussed above, which can be nonzero if the functional  region responsible for, say, movement, or some other regions, are active at the time of measurement $i$). Let us also assume that the random variable associated with samples $\bm{\gamma}_i$ are independent of the random variable associated with samples $y_i$ (this corresponds to saying that, say, the event of moving (or any brain activation event unrelated to the activation coming from the presentation of the image), is independent of the event of being presented a certain image).
Additionally, let us assume that $\| \w^* \| = 1$. 

Since $\bm{\gamma}_i$ and $\w^*$ are assumed to have disjoint support, we can therefore verify that, for all samples $i$:

\begin{equation}\label{eq:solverified}
\langle \x_i , \w^* \rangle =  y_i \|\w^*\|^2 + 0 = y_i    
\end{equation}
Therefore, $\w^*$ is indeed a solution of the system $\X \w^* = \y$

Also, we can write $\X$ as: 
$\X = \y  {\w^*}^{\top} + \bm{\Gamma}$, where each row $i$ of $\bm{\Gamma}$, seen as a column vector, is $\bm{\gamma}_i$, and based on the assumption above that every $\bm{\gamma}_i$ has a support disjoint from $\text{supp}(\w^*)$, we have: 

$$\X_S = \y {\w^*_S}^{\top} = \frac{\y}{\|\y \|}  \| \y\| {\w^*_S}^{\top}$$

Where we can recognize on the right hand side above the SVD of $\X_S$, from which we can deduce that:
 $$\X_S^{\dagger} = \w^*_S  \frac{1}{\| \y\|} \frac{\y^{\top}}{\|\y\|}$$
Which implies $$\X_S^{\dagger} \y = \frac{1}{\| \y\|^2} \w^*_S \y^{\top} \y = \w^*_{S}$$

And therefore, $\w^*_S$ is indeed here the minimum $\ell_2$ norm solution of the linear system $\X_S \w_S^* = \y$, since by property of the pseudo-inverse, such minimal $\ell_2$ norm solution is $\X_S^{\dagger} \y$.
Combined with equation \ref{eq:solverified}, we obtain that $\X$, $\y$ and $\w$ verify Assumption~\ref{ass:sol}. Now let us consider some $\ell \in S$: \\
$$\X_S^{\dagger} \x_{\ell} =  \frac{1}{\| \y\|^2} \w^*_S \y^{\top} (\y w_{\ell}) = w_{\ell} \w^*_S $$
And therefore 
\begin{equation}\label{eq:s}
    \min_{\ell \in S} | \langle \X_S^{\dagger} \x_{\ell} , \w^*_S\rangle |  =  \| \w_S^* \|^2 \min_{\ell \in S}  | w_{\ell} | > 0
\end{equation}
Where the last inequality is strictly positive, and much greater than $0$ if the smallest nonzero value of $\w^*$ is big enough (in absolute value). Additionally, on the other hand, if $\ell \in \bar{S}$, since we assumed that $\bm{\Gamma}$ is composed of variables independent of $\y$, and assuming that $\y$ and $\bm{\Gamma_{\ell}}$ have zero mean for every $\ell$, we obtain that for large enough sample size, $\y^{\top} \bm{\Gamma}_{\ell} \approx 0$, and therefore we have: 
$\X_S^{\dagger} \x_{\ell} =\frac{1}{\| \y\|^2} \w^*_S \y^{\top} \bm{\Gamma}_{\ell} \approx 0$, and therefore, 

\begin{equation}\label{eq:notS}
    \langle \X_S^{\dagger} \x_{\ell} , \w^*_S\rangle \approx 0
\end{equation}
Which therefore implies that : 
\begin{equation}
    \max_{\ell \in \bar{S}} |\langle \X_S^{\dagger} \x_{\ell} , \w^*_S \rangle | \approx 0 < \min_{\ell \in S} | \langle \X_S^{\dagger} \x_{\ell} , \w^*_S \rangle |
\end{equation}

Which is our Assumption \ref{ass:ass}, and therefore, $\X$, $\w^*$ and $\y$ verify both Assumptions \ref{ass:sol} and \ref{ass:ass} which are sufficient conditions for recovery with IRKSN.  Also note that however the matrix $\X_S$ is not injective here therefore the sufficient Assumption \ref{ass:assl1} for recovery with $\ell_1$ norm is not verified. Therefore, this might potentially explain the success of the $k$-support norm as a regularizer in fMRI tasks, contrary to $\ell_1$ norm based recovery methods which experimentally appear to produce worse results.

Finally, we emphasize that this is only a naive modeling of the true fMRI data, but we believe that it may be useful to understand the success of $k$-support norm on such particular tasks. It also gives more intuition on our conditions for recovery, and on which kind of tasks $k$-support norm may be a useful regularizer to consider.

\section{Extra experiments: prediction task on real-life datasets}\label{sec:prediction}

In this section, we run some experiments on real-life datasets to illustrate the applicability of IRKSN on prediction problems, for various datasets. Although sparse recovery is the primary goal of our paper, and is a goal distinct from prediction, we still find interesting to analyze the performance of IRKSN on predictions tasks, since those also often arise in practice.
\subsection{Setting}
As before, we consider the problem of sparse linear regression, where our goal is to minimize the expected mean squared error (MSE) loss of prediction $\mathbb{E}_{X, Y} (Y - \hat{Y})^2$, where $Y$ is the true regressed target, and $\hat{Y}$ is the predicted target, predicted linearly from the regressors $X$: 
$$\hat{Y} = \langle \hat{\bm{w}}, X\rangle + b = \sum_{i=1}^d  \hat{w}_i X_i +b  $$
(where $b$ is the intercept, fitted separately (see Appendix \ref{app:details} for more details)), and where $\hat{\bm{w}}$ is a sparse model that we seek to estimate from a training set of $n$ observations of $X$ and $Y$.
For each run, we first randomly split the data into a training set, and a test set which contains $25 \%$ of the data. Then, we split the training set into an actual training set and a validation set, with the same proportion ($75\% / 25\%$). Hyperparameters, including learning rate parameters and early stopping time are fitted to minimize the MSE on the validation set. Then the empirical MSE on the test set is reported. This procedure is repeated 10 times, and we report in Tables \ref{tab:pen1} and \ref{tab:pen2} the mean and standard deviation of that test set MSE. 

Additional details including details on the intercept and a preprocessing step, as well as the values for the grid-search of each algorithm are described in Appendix \ref{app:details}. Our results are produced on a  server of CPUs with 32 cores and 126G RAM, and take 5 hours to run.
\subsection{Datasets}
We evaluate the algorithms on the following open source datasets (obtained from the sources LibSVM \cite{libsvm} and OpenML \cite{vanschoren2014openml}), of which a brief summary is presented in Table \ref{tab:ds}.
\begin{table}[!bht]
\begin{center}
\begin{small}
\begin{sc}
  \begin{tabular}{lccc}
  \toprule
 Dataset & $d$ & $n$\\
  \midrule
     \textbf{leukemia}$^{(1)}$ & 7129 &  38   \\
  \textbf{housing} $^{(2)}$ & 8 & 20640 \\
  \textbf{scheetz2006}$^{(3)}$ & 18975 & 120 \\
    \textbf{rhee2006}$^{(4)}$ & 361 & 842 & \\
\bottomrule
\end{tabular}
\end{sc}
\end{small}
\end{center}
\caption{Datasets used in the  comparison. \textit{References:}  $^{(1)}$:  \cite{golub1999molecular}, $^{(2)}$  \cite{pace1997sparse}, $^{(3)}$:  \cite{scheetz2006regulation}, $^{(4)}$: \cite{rhee2006genotypic}. \textit{Sources:} $^{(1)}$:  \cite{libsvm}, $^{(2)}$ \cite{kooperberg1997statlib} downloaded with \texttt{scikit-learn} \cite{pedregosa2011scikit}, $^{(3, 4)}$:  \cite{breheny2022} .}
\label{tab:ds}
\end{table}
\subsection{Results}
We present our results in Tables \ref{tab:pen1} and \ref{tab:pen2}. Generally, we observe that for datasets with a large $d$ (such as \texttt{leukemia} and \texttt{scheetz2006}), $\ell_1$ based methods such as Lasso, IRCR, or SRDI achieve poorer performance: indeed, the Lasso is known to saturate when $d > n$ \cite{zou2005regularization}, i.e. its predicted $\w^*$ cannot contain more than $n$ nonzero variables. This is not the case for the ElasticNet and $k$-support norm based algorithms like IRKSN, which is why those latter algorithms achieve a good score in this $d > n$ setting. 

Perhaps surprisingly, IROSR also achieves a good score on \texttt{scheetz2006} ($d>>n$), even if its reparameterization is supposed to enforce some $\ell_1$ regularization \cite{vaskevicius2019}. However, the theory in \cite{vaskevicius2019} holds for small initializations and specific stepsizes, so we hypothesize that due to our grid search on the stepsize, our version of IRCR might be able to explore regimes beyond the $\ell_1$ norm, beyond the scope of the theory in the IRCR paper.
Our results also confirm the findings from \cite{argyriou2012sparse}, namely that the $k$-support norm regularization often outperforms the ElasticNet: this is also true for iterative regularizations using the $k$-support norm (namely, IRKSN). 

\begin{table}[!bht]
\begin{center}
\begin{small}
\begin{sc}
\begin{tabular}{llll}
  \toprule
  Method & leukemia  & housing  \\
  \midrule
 IHT & $\bm{0.322 \pm 0.137}$ & $\bm{ 0.535 \pm 0.011}$  \\
  Lasso & $0.450 \pm 0.204$& $\bm{0.535 \pm 0.016} $  \\ 
  ElasticNet & $\bm{0.307 \pm 0.154}$  &$\bm{ 0.540 \pm 0.031}$  \\ 
  KSN pen. & $\bm{{0.251 \pm 0.090}} $ & $\bm{0.533 \pm 0.009}$ \\
  OMP & $0.730 \pm 0.376$  & $ \bm{0.533 \pm 0.009}$  \\
  SRDI & $0.396 \pm 0.220$ & $\bm{0.533 \pm 0.009}$  \\
  IROSR & $0.352 \pm 0.121$  & $0.655 \pm 0.013 $  \\
  IRCR & $\bm{0.326 \pm 0.102}$  & $ \bm{0.534 \pm 0.010}$ \\
  \textbf{IRKSN (ours)} & $\bm{ 0.264 \pm 0.091}$ & $ \bm{0.538 \pm 0.012}$ \\
\bottomrule
\end{tabular}
\end{sc}
\end{small}
\end{center}
\caption{Test MSE of the methods of Table \ref{tab:pen} on the \texttt{leukemia} and \texttt{housing} datasets (bold font: mean within the standard deviation of the best score from each column).}
\label{tab:pen1}
\end{table}
\begin{table}[!bht]
\begin{center}
\begin{small}
\begin{sc}
\begin{tabular}{llll}
  \toprule
  Method &  scheetz2006 &  rhee2006 \\
  \midrule
 IHT  & $\bm{0.008 \pm 0.003} $ & $ \bm{0.576 \pm 0.053}$ \\
  Lasso & $0.012 \pm 0.008 $ & $ \bm{0.557 \pm 0.049}$ \\ 
  ElasticNet  & $\bm{0.009 \pm 0.004} $ & $\bm{0.541 \pm 0.042}$ \\ 
  KSN pen. & $\bm{0.008 \pm 0.003}$ & $\bm{0.556 \pm 0.035}$ \\
  OMP & $ 0.016 \pm 0.06 $ & $0.684 \pm 0.057$  \\
  SRDI & $0.018 \pm 0.013 $ & $\bm{0.567 \pm 0.043}$ \\
  IROSR & $\bm{0.007 \pm 0.003}$ & $\bm{0.583 \pm 0.044}$ \\
  IRCR & $ 0.018 \pm 0.013$ & $1.389 \pm 0.105$ \\
  \textbf{IRKSN (ours)} & $ \bm{0.008 \pm 0.003}$ & $ \bm{0.578 \pm 0.038} $\\
\bottomrule
\end{tabular}
\end{sc}
\end{small}
\end{center}
\caption{Test MSE of the methods of Table \ref{tab:pen} on gene array datasets (\texttt{scheetz2006} and \texttt{rhee2006}).}
\label{tab:pen2}
\end{table}

\subsection{Details on the implementation of algorithms}\label{app:details}
In this section, we present additional details on the experiments from Section \ref{sec:prediction}. First, for all the algorithms, we added a preprocessing step that centers and standardizes each column on the trainset (i.e. substract its mean and divides it by its standard deviation), and that removes columns that have 0 variance (i.e. column containing the same, replicated value). We later use this learned transformation on the validation set and the test set. In addition, we fit the intercept $b$ of the linear regression separately, as is common in sparse linear regression, by centering the target $\bm{y}$ before training, and then using the below formula for the intercept:
$$ b = \bar{y} - \langle \bar{\bm{X}}, \bm{\hat{w}}\rangle$$
Where $\bar{y}$ is the average of the target vector $\bm{y}$, $\bm{\hat{w}}$ is the final estimated model on the train set (fitted with a centered target $\bm{y} - \bar{y}$), and $\bar{\bm{X}}$ is the column-wise average of the (preprocessed) training data matrix $\bm{X}$. The prediction of a new preprocessed data sample $\bm{x'}_i$ is then $\hat{y}_i := \langle \bm{\hat{w}} , \bm{x'}_i \rangle + b$.

We recode most algorithms from scratch in \texttt{numpy} \cite{harris2020array}, except for the Lasso, ElasticNet, and OMP, for which we use the \texttt{scikit-learn} \cite{pedregosa2011scikit} implementation.  For the implementation of the proximal operator of the (half-squared) $k$-support norm (used in IRKSN and KSN penalized), we use the existing implementation from the \texttt{modopt} package \cite{farrens2020pysap}, that is based on the efficient algorithm described in \cite{mcdonald2016new}.
Below we present the grid-search parameters for each algorithms, that allowed them to achieve a good performance consistently on all datasets from Table \ref{tab:ds}.
For all iterative regularization algorithms (i.e. SRDI, IROSR, IRCR, and IRKSN), we monitor the validation MSE every 5 iterations, and choose the stopping time as the iteration number with the best MSE. We also proceed as such for IHT, since because we grid-search the learning rate, if that latter is too high, decrease of the function at each step may not be guaranteed. We run each iterative algorithm that we reimplemented (IHT, KSN penalty, SRDI, IROSR, IRCR, IRKSN) with a maximum number of iterations of 500.
Finally, we release our code at \url{https://github.com/wdevazelhes/IRKSN_AAAI2024}.

\paragraph{IHT} \cite{blumensath2009}
We search $k$ (the number of components kept at each iterations) in an evenly spaced interval from $1$ to $d$ containing 5 values, and search the learning rate $\eta$ in $\{0.0001, 0.001, 0.01, 0.1, 1.\}$.

\paragraph{Lasso} \cite{tibshirani1996} We use the implementation \texttt{lasso\_path} from \texttt{scikit-learn} \cite{pedregosa2011scikit}, with its default parameters, which automatically choses the path of  $\lambda$ based on a data criterion.

\paragraph{ElasticNet} \cite{zou2005regularization} We use the implementation \texttt{enet\_path} from \texttt{scikit-learn} \cite{pedregosa2011scikit}, which similarly as above, automatically chooses the path of  $\lambda$ based on a data criterion. In addition, we choose the recommended values $\{.1, .5, .7, .9, .95, .99, 1\}$ of \texttt{ElasticNetCV} for the relative weight of the $\ell_1$ penalty.

\paragraph{KSN penalty} \cite{argyriou2012sparse}
We choose the strenght of the $k$-support norm penalty $\lambda$ in $\{0.1, 1.\}$, the $k$ (from the $k$-support norm) in an evenly spaced interval from $1$ to $d$ containing 5 values, and we found that simply setting the constant $L$ from \cite{argyriou2012sparse} (which is the inverse of the learning rate) to $1e6$ achieves consistently good results across all datasets.

\paragraph{OMP} \cite{tropp2007}
We use the implementation from \texttt{scikit-learn} \cite{pedregosa2011scikit}, and we search $k$ in an evenly spaced interval from $1$ to $\min(n, d)$ (indeed, OMP needs $k$ not to be bigger than $\min(n, d)$) containing 5 values.
\paragraph{SRDI} \cite{osher2016}
We search for the parameters $\kappa$ and $\alpha$ from \cite{osher2016}, respectively in the intervals $\{0.0001, 0.001, 0.01, 0.1, 1.\}$ and $\{0.0001, 0.001, 0.01, 0.1, 1.\}$.
\paragraph{IROSR} \cite{vaskevicius2019}
We search for the parameters $\eta$ and $\alpha$ respectively in $\{0.0001, 0.001, 0.01, 0.1, 1.\}$ and $\{0.0001, 0.001, 0.01, 0.1, 1.\}$.

\paragraph{IRCR} \cite{molinari2021iterative} For IRSR, we found that setting $\tau$ and $\sigma$ to $\frac{0.9}{\sqrt{2 \|\bm{X}\|^2}}$ (in order to verify the condition of equation (6) in \cite{molinari2021iterative}) consistently performs well on all datasets.

\paragraph{IRKSN (ours)} For IRKSN, we search $\alpha$ (from Algorithm \ref{alg:irksn}) in $\{0.0001, 0.001, 0.01, 0.1, 1, 10\}$, and $k$ (from the $k$-support norm), in an evenly spaced interval from $1$ to $d$ containing 5 values. For the \texttt{RHEE2006} dataset, we found that the hyperparameters need to be tuned slighty more to attain comparable performance with other algorithms: the reported performance is for $\alpha=0.6$, $k=33$, ran for 1,000 iterations.

\end{document}